\documentclass[copyright, creativecommons]{eptcs}
\providecommand{\event}{GCM 2021} 
\usepackage{underscore}           

\usepackage{amsthm}

\usepackage{hyperref}

\usepackage{transparent}
\usepackage{amsmath,amssymb,amstext}

\usepackage{cite}
\usepackage{url}

\usepackage{enumitem}

\usepackage{graphicx}
\usepackage{wrapfig}

\usepackage{quiver}

\usepackage{tikz}
\usetikzlibrary{arrows,arrows.meta,automata,backgrounds,calc,chains,decorations.fractals,decorations.pathreplacing,fadings,fit,folding,mindmap,patterns,plotmarks,positioning,shadows,shapes.geometric,shapes.symbols,through,trees,hobby}
\usetikzlibrary{cd}

\colorlet{dblue}{blue!40!black}

\usepackage{enumitem}
\usepackage{verbatim}
\usepackage{float}

\usepackage{xspace}

\usepackage{latexsym}
\usepackage{fontawesome}

\usepackage{refcount}
\usepackage{caption}

\newtheorem{theorem}{Theorem}
\newtheorem{lemma}[theorem]{Lemma}
\newtheorem{corollary}[theorem]{Corollary}

\newtheorem{proposition}[theorem]{Proposition}

\theoremstyle{definition}
\newtheorem{definition}[theorem]{Definition}
\newtheorem{remark}[theorem]{Remark}
\newtheorem{example}[theorem]{Example}

\newcounter{environment}

\newcommand{\store}[4]{
  \if\relax\detokenize{#2}\relax
  \expandafter\def\csname store#3\endcsname{\begin{#1}\label{#3}#4\end{#1}}%
  \else
  \expandafter\def\csname store#3\endcsname{\begin{#1}[#2]\label{#3}#4\end{#1}}%
  \fi
  \csname store#3\endcsname%
}

\newcommand{\use}[1]{{
  \renewcommand{\label}[1]{}%
  \renewcommand{\qedappendix}{}%
  \renewcommand{\footnote}[1]{}%
  \setcounter{environment}{\arabic{theorem}}%
  \setcounterref{theorem}{#1}%
  \addtocounter{theorem}{-1}%
  \csname store#1\endcsname%
  \setcounter{theorem}{\arabic{environment}}%
}}

\newcommand{\pbpostrong}{PBPO$^{+}$\xspace}
\newcommand{\catname}[1]{{\normalfont\textbf{#1}}\xspace}

\newcommand{\Graph}{\catname{Graph}}
\newcommand{\FinGraph}{\catname{FinGraph}}
\newcommand{\GraphLattice}[1]{{\normalfont\textbf{Graph}$^{#1}$\xspace}}
\newcommand{\FinGraphLattice}[1]{{\normalfont\textbf{FinGraph}$^{#1}$\xspace}}

\newcommand{\labels}{\mathcal{L}}

\newcommand{\mono}{\rightarrowtail}

\newcommand{\iso}{\cong}

\tikzset{default/.style={
  thick,
  every node/.style={circle},
  level distance=12mm, 
  inner sep=.5mm}}
\tikzset{smallCircle/.style={circle,fill=black,inner sep=0mm,outer sep=1mm,minimum size=1mm}}
\tikzset{paint/.style={very thick,draw=#1!50!black,fill=#1,opacity=.4}}
\tikzset{paintopaque/.style={very thick,draw=#1!50!black!60,fill=#1!60}}
\tikzset{loopl/.style={out=140,in=210,looseness=10}}
\tikzset{loopr/.style={out=-40,in=30,looseness=10}}
\tikzset{loopb/.style={out=-40-90,in=30-90,looseness=10}}
\tikzset{loopt/.style={out=-40+90,in=30+90,looseness=10}}
\tikzset{loop/.style={out=-30+#1,in=30+#1,looseness=10}}
\tikzset{thinloop/.style={out=-20+#1,in=20+#1,looseness=20}}

\tikzset{emptystep/.style={-,dotted,line cap=round,dash pattern=on 0 off 3.00000}}

\tikzset{b/.style={anchor=north,at=(#1.south)}}
\tikzset{br/.style={anchor=north west,at=(#1.south east)}}
\tikzset{bl/.style={anchor=north east,at=(#1.south west)}}
\tikzset{bw/.style={anchor=north west,at=(#1.south west)}}
\tikzset{be/.style={anchor=north east,at=(#1.south east)}}
\tikzset{a/.style={anchor=south,at=(#1.north)}}
\tikzset{ar/.style={anchor=south west,at=(#1.north east)}}
\tikzset{al/.style={anchor=south east,at=(#1.north west)}}
\tikzset{aw/.style={anchor=south west,at=(#1.north west)}}
\tikzset{ae/.style={anchor=south east,at=(#1.north east)}}
\tikzset{r/.style={anchor=west,at=(#1.east)}}
\tikzset{l/.style={anchor=east,at=(#1.west)}}
\tikzset{rn/.style={anchor=north west,at=(#1.north east)}}

\tikzset{eta/.style={very thick,->,cblue!90!black}}
\tikzset{beta/.style={very thick,->,corange!80!white!90!black}}

\tikzset{devcirc/.style={circle,draw,fill=white,inner sep=0,minimum size=4\pgflinewidth}}
\tikzset{dev/.style={postaction={decorate},decoration={
  markings,
  mark=at position .5 with \node [devcirc] {};}}
}
    
\tikzset{medium tree/.style={
    level 1/.style={sibling distance=17mm},
    level 2/.style={sibling distance=9mm},
    level 3/.style={sibling distance=5mm},
    level 4/.style={sibling distance=4mm},
  }}

\tikzset{startAt/.style={inner sep=0mm,r=#1,xshift=-1.2mm,yshift=.8mm}}

\newcommand{\arrowTriangle}[2]{
  \pgfpathmoveto{\pgfpoint{-0.01#1+#2}{.6#1}}
  \pgfpathlineto{\pgfpoint{#1+#2}{0}}
  \pgfpathlineto{\pgfpoint{-0.01#1+#2}{-.6#1}}
  \pgfusepathqfill
}

\newdimen\prearrowsize
\newdimen\arrowsize
\newdimen\temparrowsize
\newcommand{\arrowscale}{5}
\newcommand{\setarrowsize}{
  \arrowsize=0.000000001pt
  \prearrowsize=\arrowscale\pgflinewidth
  \normalizearrowsize
}
\newcommand{\normalizearrowsize}{
  \ifdim\prearrowsize>2mm
    \addtolength{\arrowsize}{2mm}
    \addtolength{\prearrowsize}{-2mm}
    \temparrowsize=0.5\prearrowsize
    \prearrowsize=\temparrowsize
  \else
  \fi

  \addtolength{\arrowsize}{\prearrowsize}
}

\pgfarrowsdeclarealias{share>}{share>}{stealth'}{stealth'}%

\pgfarrowsdeclare{red>}{red>}
{
  \setarrowsize
  \pgfarrowsleftextend{1.5\pgflinewidth} 
  \pgfarrowsrightextend{\arrowsize-0.1\arrowsize}
}{
  \setarrowsize
  \pgfsetdash{}{0pt} 
  \pgfsetroundjoin 
  \pgfsetroundcap 
  \arrowTriangle{\arrowsize}{-0.1\arrowsize}
}

\pgfarrowsdeclare{red>>}{red>>}
{
  \setarrowsize
  \pgfarrowsleftextend{1.5\pgflinewidth}
  \pgfarrowsrightextend{\arrowsize-0.1\arrowsize+.8\arrowsize}
}{
  \setarrowsize
  \pgfsetdash{}{0pt} 
  \pgfsetroundjoin 
  \pgfsetroundcap 
  \arrowTriangle{\arrowsize}{-0.1\arrowsize}
  \arrowTriangle{\arrowsize}{-0.1\arrowsize+.8\arrowsize}
}

\pgfarrowsdeclare{red>>>}{red>>>}
{
  \setarrowsize
  \pgfarrowsleftextend{1.5\pgflinewidth}
  \pgfarrowsrightextend{\arrowsize-0.1\arrowsize+1.6\arrowsize}
}{
  \setarrowsize
  \pgfsetdash{}{0pt} 
  \pgfsetroundjoin 
  \pgfsetroundcap 
  \arrowTriangle{\arrowsize}{-0.1\arrowsize}
  \arrowTriangle{\arrowsize}{-0.1\arrowsize+.8\arrowsize}
  \arrowTriangle{\arrowsize}{-0.1\arrowsize+1.6\arrowsize}
}

\tikzstyle{gyellow}=[draw=black!80,top color=white!50,bottom color=black!20]
\tikzstyle{gblue}=[draw=blue!50,top color=white,bottom color=blue!60]
\tikzstyle{gred}=[draw=red!50,top color=white,bottom color=red!60]
\tikzstyle{ggreen}=[draw=blue!80!green!90!black,top color=white,bottom color=blue!80!green!60]

\tikzstyle{roundNode}=[gyellow,thick,circle,minimum size=4mm,inner sep=0.5mm]

\definecolor{cblue}{rgb}{0,0.4,0.7}
\definecolor{clighterblue}{rgb}{0,0.6,1.0}
\colorlet{cred}{red}
\colorlet{cgreen}{green!80!black}
\colorlet{corange}{orange!70!red}
\colorlet{cpureorange}{orange}
\colorlet{cpurple}{clighterblue!50!cred}

\colorlet{clightblue}{clighterblue!50!cblue!40}
\colorlet{clightred}{cred!40}
\colorlet{clightgreen}{cgreen!80!cblue!40}
\colorlet{clightyellow}{corange!40!yellow!50}
\colorlet{clightorange}{cred!50!orange!40}
\colorlet{clightpurple}{clighterblue!50!cred!50}

\colorlet{cdarkred}{cred!70!black}
\colorlet{cdarkgreen}{cgreen!60!black}
\colorlet{cdarkblue}{cblue!60!black}

\colorlet{chighlight}{orange!50!yellow!60}

\tikzset{pgnode/.style={smallCircle,fill=white,draw=black,minimum size=1.3mm,outer sep=0.5mm}}
\tikzset{pgnodecolor/.style={pgnode,minimum size=4mm,scale=0.9}}
\tikzset{pgnodebig/.style={roundNode,gyellow,outer sep=1mm}}
\tikzset{pgrelation/.style={ultra thick,cblue!80!black,decorate,decoration={snake,amplitude=.4mm,segment length=4mm}}}

\tikzset{exi/.style={densely dotted}}
\tikzset{decweak/.style={-,green!50!black,very thick,opacity=0.7}}
\tikzset{decstrict/.style={->,orange,very thick,opacity=0.7}}

\tikzset{graphNode/.style={circle,draw=black,inner sep=.5mm,outer sep=1mm}}

\tikzset{sloop/.style={looseness=7}}

\tikzset{interconnect/.style={dotted,cred}}
\tikzset{match/.style={cdarkgreen,very thick}}
\tikzset{jigsaw/.style={circle,draw=black,minimum size=2mm,fill=white,minimum size=3.5mm}}

\colorlet{myblue}{blue!80!black}
\colorlet{mygreen}{cdarkgreen}
\colorlet{myred}{cred}
\colorlet{myorange}{corange}
\colorlet{mypurple}{blue!40!cred}

\tikzset{smalljigsaw/.style={rectangle,rounded corners=3mm,inner sep=1mm}}

\tikzset{epattern/.style={}}
\tikzset{eset/.style={draw=black!50}}
\tikzset{npattern/.style={rectangle,rounded corners=2mm,draw=black,inner sep=0.5mm,outer sep=.5mm,minimum size=4.5mm}}
\tikzset{nset/.style={npattern,draw=black!50,fill=white}}
\tikzset{label/.style={scale=0.85,inner sep=0,outer sep=0.5mm}}

\tikzset{short/.style={node distance=10mm}}

\newcommand{\annotate}[2]{
  \node at (#1.north east) [anchor=south west,xshift=-.75mm,yshift=-.75mm,label,outer sep=0mm] {#2};
}

\newcommand{\graphbox}[8]{
  \begin{scope}[xshift=#2,yshift=#3]
    \draw [rounded corners=2mm] (0,0) rectangle (#4,-#5);
    \node at (0,0mm) [anchor=north west,inner sep=1mm] {#1};
    \begin{scope}[xshift=#4/2+#6,yshift=#7] 
    #8
    \end{scope}
  \end{scope}
}

\newcommand{\transparentgraphbox}[8]{
{
  \transparent{0.5}
  \graphbox{#1}{#2}{#3}{#4}{#5}{#6}{#7}{#8}
  }
}

\newcommand{\vertex}[2]{%
  \begin{tikzpicture}[baseline=-1ex]%
    \node [rectangle,rounded corners=2mm,inner sep=0.5mm,fill=#2] {$#1$};%
  \end{tikzpicture}%
}

\newcommand{\rulescale}{0.9}

\makeatletter

\newcommand*{\saveparinfo}[1]{%
  \expandafter\xdef\csname my@#1@parindent\endcsname{\the\parindent}%
  \expandafter\xdef\csname my@#1@parskip\endcsname{\the\parskip}%
}

\newcommand*{\useparinfo}[1]{%
  \parindent=\csname my@#1@parindent\endcsname \relax
  \parskip=\csname my@#1@parskip\endcsname \relax
}

\makeatother

\newcommand{\qedappendix}{\hfill \textcolor{blue}{$\circledast$}}

\usepackage{calrsfs}
\DeclareMathAlphabet{\pazocal}{OMS}{zplm}{m}{n}
\renewcommand{\mathcal}{\pazocal}

\tikzset{graphNode/.style={circle,draw=black,inner sep=.5mm,outer sep=1mm}}

\begin{document}

\title{From Linear Term Rewriting to Graph Rewriting \\ with Preservation of Termination}


\author{
Roy Overbeek \qquad\qquad J\"{o}rg Endrullis
\institute{Vrije Universiteit Amsterdam\\
Amsterdam, The Netherlands}
\email{r.overbeek@vu.nl \qquad j.endrullis@vu.nl}
}
\def\titlerunning{From Linear Term Rewriting to Graph Rewriting with Preservation of Termination}
\def\authorrunning{R. Overbeek \& J. Endrullis}

\maketitle

\begin{abstract}%
Encodings of term rewriting systems (TRSs) into graph rewriting systems usually lose global termination, meaning the encodings do not terminate on all graphs. A typical encoding of the terminating TRS rule $a(b(x)) \to b(a(x))$, for example, may be indefinitely applicable along a cycle of $a$'s and $b$'s. 
Recently, we introduced \pbpostrong, a graph rewriting formalism in which rules employ a type graph to specify transformations and control rule applicability. In the present paper, we show that \pbpostrong allows for a natural encoding of linear TRS rules that preserves termination globally.
This result is a step towards modeling other rewriting formalisms, such as lambda calculus and higher order rewriting, using graph rewriting in a way that preserves properties like termination and confluence. We moreover expect that the encoding can serve as a guide for lifting TRS termination methods to \pbpostrong rewriting.
\end{abstract}

\saveparinfo{parinfo}%

\section{Introduction}
\label{sec:introduction}

\newcommand{\framework}[1]{{\mathcal{#1}}}
\newcommand{\encoding}{\mathcal{E}}

A \emph{rewriting framework} $\framework{F}$ consists of a set of \emph{objects} $O$ and a set of \emph{rewriting systems} $\mathcal{R}$. Each system $R \in \mathcal{R}$ is a set of \emph{rewrite rules}. Each rule $\rho \in R$ defines a particular \emph{rewrite relation} ${\to_\rho} \subseteq {O \times O}$ on objects, and the rules of $R$ collectively give rise to a general rewrite relation ${\to_R} = {\bigcup_{\rho \in R} {\to_\rho}}$. The usual definitions of string, cycle and term rewriting systems (TRSs), and the various definitions of term graph and graph rewriting formalisms, are instances of this abstract view.

Because terms can be viewed as generalizations of strings, term graphs as generalizations of terms, graphs as generalizations of terms graphs and cycles, etc., the question whether one framework can be encoded into another framework frequently arises naturally. The same is true when comparing the large variety of graph rewriting frameworks. Moreover, the properties such an encoding is expected to satisfy may vary. Let us therefore fix some vocabulary.

\begin{definition}[Encoding]
\label{def:encoding}
An \emph{encoding} $\encoding$ of a framework $\framework{F}$ into a framework $\framework{G}$ consists of a function $\encoding_O : O^\framework{F} \to O^\framework{G}$ on objects and a function $\encoding_\mathcal{R} : \mathcal{R}^\framework{F} \to \mathcal{R}^\framework{G}$ on rewrite systems. The subscript is usually omitted, since it will be clear from context which of $\encoding_O$ and $\encoding_\mathcal{R}$ is meant.

Given an encoding $\encoding$, a variety of properties of interest may be distinguished. We will say $\encoding$ is 
\begin{enumerate}
    \item \emph{step-preserving} if ${x \to_R^\framework{F} y} \; \Longrightarrow \; {\encoding(x) \to_{\encoding(R)}^\framework{G} \encoding(y)}$;
    \item \emph{closed} if 
    $
    {x \to_{\encoding(R)}^\framework{G} y} \text{ and } {x \iso \encoding(x') \text{ for an $x' \in O^\framework{F}$}}
    \;
    \Longrightarrow
    \;
     {y \iso \encoding(y') \text{ for some $y' \in O^\framework{F}$}} \text{ with } {x' \to_R y'}
     $;
    \item an \emph{embedding} if $\encoding$ is step-preserving and closed;
    \item \emph{globally $P$-preserving} (for a property $P$, such as termination or confluence), if whenever $R \in \mathcal{R}^\framework{F}$ satisfies  $P$, then so does the system
    $\encoding(R) \in \mathcal{R}^\framework{G}$ on all objects $O^\framework{G}$; and
    \item \emph{locally $P$-preserving}
    if whenever $R \in \mathcal{R}^\framework{F}$ satisfies  $P$, then so does the system
    $\encoding(R) \in \mathcal{R}^\framework{G}$ on the restricted domain of objects $ \encoding(O^\framework{F}) \subseteq O^\framework{G}$.
\end{enumerate}
\end{definition}

Consider the string rewrite rule $ab \to ba$ and its usual encoding $a(b(x)) \to b(a(x))$ as a term rewrite rule. This encoding is an embedding that preserves termination and confluence globally. The usual encoding as a cycle rewrite rule, by contrast, is step-preserving, but not closed, and neither termination- nor confluence-preserving.

Building on PBPO by Corradini et al.~\cite{corradini2019pbpo} and our own patch graph rewriting formalism~\cite{overbeek2020patch}, we recently proposed the \pbpostrong algebraic graph rewriting approach~\cite{overbeek2021pbpo}, in which rules employ a type graph to specify transformations and control rule applicability.
In the present paper we give an embedding of linear term rewrite systems into \pbpostrong{} that preserves global termination, despite being applicable to graphs that are not encodings of terms. This result requires powerful features (unsupported by standard approaches such as DPO~\cite{ehrig1973graph}), as two examples illustrate:
\begin{enumerate}
    \item For the encoding of $f(x,y) \to f(a,y)$ to be step-preserving, it must be possible to delete an arbitrary subgraph $x$ below $f$, while leaving
    the context above of $f$ and the subgraph corresponding to $y$ intact.
    \item For the encoding of $a(b(x)) \to b(a(x))$ to be terminating, the rule must not be applicable on a cycle of $a$'s and $b$'s.
\end{enumerate}

Apart from being an interesting expressiveness result for \pbpostrong, our result enables reduction-style termination arguments for linear, `term-like' \pbpostrong rewrite rules. Moreover, as we will elaborate in the discussion (Section~\ref{sec:discussion}), we believe our result has broader relevance for the development of termination techniques for graph rewriting, as well as the modeling of other rewrite formalisms.

The structure of the paper is as follows. In Section~\ref{sec:preliminaries}, we summarize the relevant categorical and TRS preliminaries. In Section~\ref{sec:pbpostrong}, we give a self-contained introduction to \pbpostrong and to
\GraphLattice{(\labels,\leq)}~\cite{overbeek2021pbpo}, a special category that combines well with \pbpostrong. In Section~\ref{sec:embedding:rewriting:systems}, we define an embedding of linear term rewriting into \pbpostrong rewriting over category \GraphLattice{(\labels,\leq)}. In Section~\ref{sec:termination}, we prove that the embedding is globally termination-preserving, using a novel zoning proof. Finally, we discuss the significance of our results in~Section~\ref{sec:discussion}.

\section{Preliminaries}
\label{sec:preliminaries}

We assume familiarity with various basic categorical notions, notations and results, including morphisms $X \to Y$, pullbacks and pushouts, monomorphisms (monos) $X \mono Y$ (note the different arrow notation) and identities $1_X : X \mono X$~\cite{mac1971categories,awodey2006category}.

\newcommand{\src}{\mathit{s}}
\newcommand{\tgt}{\mathit{t}}
\newcommand{\lbl}{\ell}
\newcommand{\lblv}{\lbl^V}
\newcommand{\lble}{\lbl^E}

\begin{definition}[Graph Notions]
    A (labeled) \emph{graph} $G$ consists of a set of vertices $V$, a set of edges $E$, source and target functions $\src,\tgt : E \to V$, and label functions $\lblv : V \to \labels$ and $\lble : E \to \labels$ for some label set $\labels$.
    A graph is \emph{unlabeled} if $\labels$ is a singleton.
    
    A \emph{premorphism} between graphs $G$ and $G'$ is a pair of maps $\phi = (\phi_V : V_G \to V_{G'}, \phi_E : E_G \to E_{G'})$ satisfying
    $(s_{G'}, t_{G'}) \circ \phi_E = \phi_V \circ (s_G, t_G)$.
    
    A \emph{homomorphism} is a label-preserving premorphism $\phi$, i.e., a premorphism satisfying $\lblv_{G'} \circ \phi_V = \lblv_{G}$ and $\lble_{G'} \circ \phi_E = \lble_{G}$.
\end{definition}

\begin{definition}[Category $\Graph$~\cite{ehrig2006}]
    The category $\Graph$ has graphs as objects, parameterized over some global (and usually implicit) label set $\labels$, and homomorphisms as arrows. $\FinGraph$ is the full subcategory of finite graphs.
\end{definition}

\newcommand{\xvar}{\mathcal{X}}
\newcommand{\ter}{\mathrm{Ter}(\Sigma, \xvar)}
\newcommand{\hole}{\Box}
\newcommand{\con}{\mathrm{Ter}(\Sigma, \xvar \uplus \{\,\hole\,\})}

The following TRS definitions are all standard~\cite{terese}.
\begin{definition}[Signature]
    A \emph{signature} $\Sigma$ consists of a non-empty set of \emph{function symbols} $f,g,\ldots \in \Sigma$, equipped with an \emph{arity function} $\# : \Sigma \to \mathbb{N}$. Nullary function symbols $a,b,\ldots$ are called \emph{constants}.
\end{definition}

\newcommand{\term}[1]{\mathbf{#1}}
\newcommand{\var}[1]{\mathrm{Var}(#1)}

\begin{definition}[Terms]
    The set of \emph{terms} $\term{l},\term{r}, \term{s},\term{t},\ldots \in \ter$ over a signature $\Sigma$ and an infinite set of variables $x,y,\ldots \in \xvar$ is defined inductively by:
    \begin{itemize}
        \item $x \in \ter$ for every $x \in \xvar$;
        \item if $f \in \Sigma$ with $\#(f) = n$, and $\term{t_1}, \ldots, \term{t_n} \in \ter$, then $f(\term{t_1}, \ldots, \term{t_n}) \in \ter$. If $n = 0$, we write $f$ instead of $f()$.
    \end{itemize}
    A term $\term{t}$ is \emph{linear} if every $x \in \xvar$ occurs at most once in $\term{t}$. We write $\var{\term{t}}$ to denote the set of variables occurring in $\term{t}$.
\end{definition}
 
\newcommand{\nat}{\mathbb{N}}
\newcommand{\pos}[1]{\textit{Pos}(#1)}
\newcommand{\posof}[2]{\textit{Pos}^{#1}(#2)}
\begin{definition}[Position]
\label{def:position}
  A \emph{position} $p$ is a sequence of integers, i.e., $p \in \nat^*$.
  The empty sequence is denoted by $\epsilon$. We write $pn$ (and $np$) to denote the right (and left) concatenation of a positive integer $n$ to a position $p$.
  
  \newcommand{\symbolatpos}[2]{{#1}(#2)}
  Every symbol occurrence in a term has a position associated with it. The position of the head symbol is $\epsilon$, and the position of the $i$-th ($i \geq 1$) symbol below a symbol with position $p$ is $pi$.
  For a term $\term{s}$ and a position $p$ in $s$, we write
  $\symbolatpos{\term{s}}{p}$ to denote the symbol at position $p$ in~$\term{s}$.
\end{definition}
 
\begin{definition}[Substitutions]
  A \emph{substitution} is a function $\sigma : \xvar \to \ter$. 
  For terms $\term{s} \in \ter$ we define $\term{s}\sigma \in \ter$ by 
  $x \sigma = \sigma(x)$ for $x \in \xvar$,
  and 
  $f(\term{t_1},\ldots,\term{t_n}) \sigma = f(\term{t_1}\sigma,\ldots,\term{t_n}\sigma)$ 
  for $f \in \Sigma$ and $\term{t_1},\ldots,\term{t_n} \in \ter$.
\end{definition}

\begin{definition}[Contexts]
  A \emph{context} $C[\;]$ is a term from $\con$ with exactly one occurrence of the \emph{hole} $\hole$. We write $C[\term{t}]$ for the term obtained by replacing the hole with $\term{t}$.
\end{definition}
 
\begin{definition}[Term Rewriting Systems]
  A \emph{term rewrite rule} is a pair of terms $\term{l} \to \term{r}$ satisfying $\term{l} \notin \xvar$ and $\var{\term{r}} \subseteq \var{\term{l}}$. The rule is \emph{linear} if both terms $\term{l}$ and $\term{r}$ are linear.
  A \emph{term rewriting system}  (\emph{TRS}) $\mathcal{R}$ is a set of term rewrite rules.
  The system $\mathcal{R}$ is linear if all its rules are.

  A TRS $\mathcal{R}$ induces a relation ${\to}$ on $\ter$, the \emph{rewrite relation of~$\mathcal{R}$}, as follows:
  $C[\term{l}\sigma] \to C[\term{r}\sigma]$ for every context $C$, substitution $\sigma$ and rule ${\term{l} \to \term{r}} \in \mathcal{R}$.
  The \emph{rewrite step} $C[\term{l}\sigma] \to C[\term{r}\sigma]$ is said to be an \emph{application of the rule $\rho = {\term{l} \to \term{r}}$ at position $p$}, where $p$ is the position of the hole in $C[\;]$.
\end{definition}

\section{\pbpostrong and \GraphLattice{(\labels,\leq)}}
\label{sec:pbpostrong}

We recently introduced \pbpostrong~\cite{overbeek2021pbpo} (short for \emph{PBPO with strong matching}), an algebraic rewriting formalism obtained by strengthening the matching mechanism of PBPO by Corradini et al.~\cite{corradini2019pbpo}. 
We believe \pbpostrong is of interest for at least three important reasons.

First, \pbpostrong is expressive: for $\Graph$ in particular, and assuming monic matching, we conjecture~\cite{overbeek2021pbpo} that \pbpostrong is able to faithfully model DPO, SPO~\cite{lowe1993algebraic}, SqPO\cite{corradini2006sesqui}, AGREE~\cite{corradini2015agree} and PBPO. 
More precisely, for any rule in such a formalism, there exists a \pbpostrong rule that generates exactly the same rewrite relation.

Second, \pbpostrong makes relatively weak assumptions on the underlying category: it is sufficient to require the existence of pushouts along monomorphisms and the existence of pullbacks. In particular,  adhesivity~\cite{lack2004adhesive}, assumed for DPO rewriting to ensure the uniqueness of pushout complements, is not required.

Third, we have defined a non-adhesive category called   \GraphLattice{(\labels,\leq)}~\cite{overbeek2021pbpo} that combines very nicely with \pbpostrong, allowing graph rewrite rules to easily model notions of relabeling, type systems, wildcards and variables. These notions have been significantly more challenging to define for DPO.

In this section we provide the necessary background on  \pbpostrong and
\GraphLattice{(\labels,\leq)}.

\newcommand{\picturerule}{
  \begin{tikzpicture}[node distance=11mm,l/.style={inner sep=.5mm},baseline=-6mm]
    \node (L) {$L$};
    \node (K) [right of=L] {$K$}; \draw [->] (K) to node [above] {$l$} (L);
    \node (LP) [below of=L] {$L'$};
      \draw [>->] (L) to node [left] {$t_L$} (LP);
    \node (KP) [below of=K] {$K'$};
      \draw [>->] (K) to node [right] {$t_K$} (KP);
      \draw [->] (KP) to node [below] {$l'$} (LP);
       \node at ([shift={(-4mm,-4mm)}]K) {PB};
    \node (R) [right of=K] {$R$}; 
      \draw [->] (K) to node [above] {$r$} (R);
  \end{tikzpicture}
}

\newcommand{\picturestep}{
  \begin{tikzpicture}[node distance=16mm,l/.style={inner sep=1mm},baseline=-7.5mm,v/.style={node distance=11mm}]
    \node (G) {$G_L$};
    \node (L) [left of=G] {$L$};
      \draw [>->] (L) to node [above,l] {$m$} (G);
    \node (L2) [below of=L,v] {$L$};
      \draw [>->] (L) to node [left,l] {$1_L$} (L2);
      \node at ([shift={(4mm,-4mm)}]L) {PB};
    \node (LP) [below of=G,v] {$L'$};
      \draw [>->] (L2) to node [above,l] {$t_L$} (LP);
      \draw [->] (G) to node [left,l] {$\alpha$} (LP);
     \node (GKP) [right of=G] {$G_K$};
     \draw [->] (GKP) to node [above,l] {$g_L$} (G);
     \node (KP) [below of=GKP,v] {$K'$};
     \draw [->] (GKP) to node [right,l] {$u'$} (KP);
     \draw [->] (KP) to node [above,l] {$l'$} (LP);
     \node at ([shift={(-4mm,-4mm)}]GKP) {PB};
     \node (K) [above of=GKP,v] {$K$};
     \draw [>->, dotted] (K) to node [left,l] {$!u$} (GKP);
     \node (R) [right of=K] {$R$};
     \draw [->] (K) to node [above,l] {$r$} (R);
     \node (GP) [below of=R,v] {$G_R$};
     \draw [->] (GKP) to node [below,l,pos=0.7] {$g_R$} (GP);
     \draw [->] (R) to node [right,l] {$w$} (GP);
     \node at ([shift={(-4mm,4mm)}]GP) {PO};
     
     
     \draw [>->] (K) to[out=-45,in=30,looseness=0.8] node[right,l,pos=0.85] {$t_K$} (KP);
  \end{tikzpicture}
}

\begin{definition}[\pbpostrong Rewriting~\cite{overbeek2021pbpo}] \label{def:pbpostrong:rewrite:step}
    A \emph{\pbpostrong rewrite rule} $\rho$ (left)
    and \emph{adherence morphism} $\alpha : G_L \to L'$ induce a \emph{rewrite step} $G_L \Rightarrow_\rho^\alpha G_R$ on arbitrary $G_L$ and $G_R$ if the properties indicated by the commuting diagram on the right hold
    \begin{center}
      \raisebox{9mm}{$\rho \ = \ $\picturerule} \hspace{1cm} \picturestep
    \end{center}
    where $u : K \to G_K$ is the unique mono satisfying $t_K = u' \circ u$ \cite[Lemma 11]{overbeek2021pbpo}. We write $G_L \Rightarrow_\rho G_R$ if $G_L \Rightarrow_\rho^\alpha G_R$ for some $\alpha$.
\end{definition}

In the rewrite rule diagram, $L$ is the \emph{lhs pattern} of the rule, $L'$ its \emph{type graph} and $t_L$ the \emph{typing} of $L$. Similarly for the \emph{interface} $K$. $R$ is the \emph{rhs pattern} or \emph{replacement for $L$}. The rewrite step diagram can be thought of as consisting of a match square (modeling an application condition), a pullback square for extracting (and possibly duplicating) parts of $G_L$, and finally a pushout square for gluing these parts along pattern $R$. The inclusion of the match square is the main aspect which differentiates \pbpostrong from PBPO: intuitively, it prevents $\alpha$ from collapsing context elements of $G_L$ onto the pattern $t_L(L) \subseteq L'$.

For the present paper, it suffices to restrict attention to rules in which $l'$ does not duplicate subgraphs.

\begin{definition}[Linear \pbpostrong Rule]
A \pbpostrong rule is \emph{linear} if the morphism $l' : K' \to L'$ is monic.
\end{definition}

\begin{remark}
For linear \pbpostrong rewriting, it is enough to assume the existence of pushouts and pullbacks along monomorphisms. An interesting question is whether these weakened requirements enable new use cases.
\end{remark}

The category \GraphLattice{(\labels,\leq)} is similar to $\Graph$. The difference is that it is assumed that the label set forms a complete lattice, and that morphisms do not decrease labels. The complete lattice requirement ensures that pushouts and pullbacks are well-defined.

\begin{definition}[Complete Lattice]
A \emph{complete lattice} $(\labels, \leq)$ is a poset such that all subsets $S$ of $\labels$ have a supremum (join) $\bigvee S$ and an infimum (meet) $\bigwedge S$. 
\end{definition}

\begin{definition}[Category \GraphLattice{(\labels,\leq)}~\cite{overbeek2021pbpo}]
For a complete lattice $(\labels, \leq)$, the category \GraphLattice{(\labels,\leq)} is the category in which objects are graphs are labeled from $\labels$, and arrows are graph premorphisms $\phi : G \to G'$ that satisfy $\lbl_G(x) \leq \lbl_{G'}(\phi(x))$ for all $x \in V_G \cup E_G$. We let \FinGraphLattice{(\labels,\leq)} denote the full subcategory of finite graphs.
\end{definition}

\begin{proposition}
\label{prop:mono:stable:pushout}
In {\GraphLattice{(\labels,\leq)}}, monomorphisms are stable under pushout.
\end{proposition}
\begin{proof}
Assume given a span $B \stackrel{b}{\leftarrow} A \stackrel{c}{\mono} C$ in \GraphLattice{(\labels,\leq)}. Overloading names, consider the unlabeled version in \Graph, and construct the pushout $B \stackrel{m}{\to} D \stackrel{n}{\leftarrow} C$. Morphism $m : B \to D$ is monic, because monos are stable in the category of unlabeled graphs, by virtue of it being an adhesive category. Now for each $x \in V_D \cup E_D$, define the label function $\lbl(x)$ to be the supremum of all labels in the labeled preimages $m^{{-1}}(x)$ and $n^{{-1}}(x)$, and define the \GraphLattice{(\labels,\leq)} object $D_\lbl = (V_D, E_D, s_D, t_D, \lbl)$. Then it is easy to verify that $B \stackrel{m}{\mono} D_\lbl \stackrel{n}{\leftarrow} C$ is the pushout of 
$B \stackrel{b}{\leftarrow} A \stackrel{c}{\mono} C$ in \GraphLattice{(\labels,\leq)}.
\end{proof}

In this paper we will use the following simple complete lattice only.

\begin{definition}[Flat Lattice~\cite{overbeek2021pbpo}]
Let $\labels^{\bot,\top} = \labels \uplus \{ \bot, \top \}$. We define the \emph{flat lattice} induced by $\labels$ as the poset 
$(\labels^{\bot,\top}, {\leq})$, which has $\bot$ as a global minimum and $\top$ as a global maximum, and where all elements of $\labels$ are incomparable. In this context, we refer to $\labels$ as the \emph{base label set}.
\end{definition}

The following example is a variation of an example found in our previous paper~\cite[Example 40]{overbeek2021pbpo}. It exemplifies all relevant features of linear \pbpostrong rewriting in category \GraphLattice{(\labels,\leq)}.

\begin{example}[Rewrite Example]
\label{example:hard:overwriting}
  As vertex labels we employ the flat lattice induced by the base label set $\{\, a,b,c,\ldots \,\}$, and we assume edges are unlabeled for notational simplicity.
  The diagram
  {

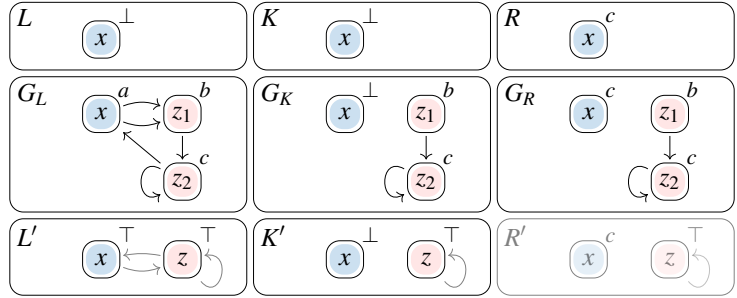
\newcommand{\nodexa}{\vertex{x_1}{cblue!20}}
\newcommand{\nodexb}{\vertex{x_2}{cblue!20}}
\newcommand{\nodexc}{\vertex{x_3}{cblue!20}}
\newcommand{\nodexd}{\vertex{x_4}{cblue!20}}

\newcommand{\nodea}{\vertex{a}{cgreen!20}}
\newcommand{\nodeb}{\vertex{b}{cpurple!25}}

\newcommand{\nodeaa}{\vertex{a_1}{cgreen!20}}
\newcommand{\nodeab}{\vertex{a_2}{cgreen!20}}
\newcommand{\nodeba}{\vertex{b_1}{cpurple!25}}
\newcommand{\nodebb}{\vertex{b_2}{cpurple!25}}

\newcommand{\nodex}{\vertex{x}{cblue!20}}

\newcommand{\nodez}{\vertex{z}{cred!10}}
\newcommand{\nodeza}{\vertex{z_1}{cred!10}}
\newcommand{\nodezb}{\vertex{z_2}{cred!10}}

\begin{center}
  \scalebox{\rulescale}{
  \begin{tikzpicture}[->,node distance=12mm,n/.style={}]
    \graphbox{$L$}{0mm}{0mm}{35mm}{10mm}{-4mm}{-5.5mm}{
      \node [npattern] (x)
      {\nodex};
       \annotate{x}{$\bot$};
    }
    \graphbox{$K$}{36mm}{0mm}{35mm}{10mm}{-4mm}{-5.5mm}{
      \node [npattern] (x)
      {\nodex};
      \annotate{x}{$\bot$};
    }
    \graphbox{$R$}{72mm}{0mm}{35mm}{10mm}{-4mm}{-5.5mm}{
      \node [npattern] (x)
      {\nodex};
       \annotate{x}{$c$};
    }
    \graphbox{$G_L$}{0mm}{-11mm}{35mm}{20mm}{-4mm}{-5.5mm}{
      \node [npattern] (x) {\nodex};
        \annotate{x}{$a$};
      \node [npattern] (za) [right of=x] {\nodeza};
        \annotate{za}{$b$};
      \node [npattern] (zb) [short, below of=za] {\nodezb};
        \annotate{zb}{$c$};
        
        \draw [epattern] (x) to [bend left=20] node {} (za);
        \draw [epattern] (x) to [bend right=20] node {} (za);
        \draw [epattern] (za) to node {} (zb);
        \draw [epattern] (zb) to node {} (x);
        
        \draw [epattern,loop=-180,looseness=3] (zb) to node {} (zb);
    }
    \graphbox{$G_K$}{36mm}{-11mm}{35mm}{20mm}{-4mm}{-5.5mm}{
      \node [npattern] (x) {\nodex};
        \annotate{x}{$\bot$};
      \node [npattern] (za) [right of=x] {\nodeza};
        \annotate{za}{$b$};
      \node [npattern] (zb) [short, below of=za] {\nodezb};
        \annotate{zb}{$c$};
        
        \draw [epattern] (za) to node {} (zb);
        
        \draw [epattern,loop=-180,looseness=3] (zb) to node {} (zb);
    }
    \graphbox{$G_R$}{72mm}{-11mm}{35mm}{20mm}{-4mm}{-5.5mm}{
      \node [npattern] (x) {\nodex};
        \annotate{x}{$c$};
      \node [npattern] (za) [right of=x] {\nodeza};
        \annotate{za}{$b$};
      \node [npattern] (zb) [short, below of=za] {\nodezb};
        \annotate{zb}{$c$};
        
        \draw [epattern] (za) to node {} (zb);
        
        \draw [epattern,loop=-180,looseness=3] (zb) to node {} (zb);
    }
    \graphbox{$L'$}{0mm}{-32mm}{35mm}{11.3mm}{-4mm}{-6mm}{
      \node [npattern] (x)
      {\nodex};
       \annotate{x}{$\top$};
        \node [npattern] (z) [right of=x] {\nodez};
       \annotate{z}{$\top$};
        \draw [eset] (x) to [bend right=20] node {} (z);
         \draw [eset] (z) to [bend right=10] node {} (x);
         \draw [eset,loop=-20,looseness=3] (z) to node {} (z);
    }
    \graphbox{$K'$}{36mm}{-32mm}{35mm}{11.3mm}{-4mm}{-6mm}{
      \node [npattern] (x)
      {\nodex};
       \annotate{x}{$\bot$};
        \node [npattern] (z) [right of=x] {\nodez};
       \annotate{z}{$\top$};
         \draw [eset,loop=-20,looseness=3] (z) to node {} (z);
    }
    \transparentgraphbox{$R'$}{72mm}{-32mm}{35mm}{11.3mm}{-4mm}{-6mm}{
      \node [npattern] (x)
      {\nodex};
       \annotate{x}{$c$};
        \node [npattern] (z) [right of=x] {\nodez};
       \annotate{z}{$\top$};
         \draw [eset,loop=-20,looseness=3] (z) to node {} (z);
    }
  \end{tikzpicture}
  }
\end{center}

}

  \noindent
  displays a rule ($L,L',K,K',R$) which
  \begin{itemize}
      \item matches an arbitrarily labeled, loopless node $x$, in an arbitrary context;
      \item ``hard overwrites'' the label of $x$ to label $c$;
      \item disconnects $x$ from its component by deleting its incident edges; and
      \item leaves all other nodes, edges and labels unchanged.
  \end{itemize}
  The pushout $K' \xrightarrow{r'} R' \xleftarrow{t_R} R$ for span $K' \xleftarrow{t_K} K \xrightarrow{r} R$ is depicted as well (in lower opacity), because it shows the schematic effect of applying the rewrite rule. An application to a host graph $G_L$ is included in the middle row.

  With respect to the labeling, the example demonstrates how (i)~labels in $L$ serve as lower bounds for matching, (ii)~labels in $L'$ serve as upper bounds for matching, (iii)~labels in $K'$ can be used to decrease matched labels (so in particular, $\bot$ ``instructs'' to ``erase'' the label and overwrite it with $\bot$, and $\top$ ``instructs'' to preserve labels), and (iv)~labels in $R$ can be used to increase labels.
\end{example}

\section{Embedding Linear Term Rewriting Systems}
\label{sec:embedding:rewriting:systems}

We are now ready to define an encoding (Definition~\ref{def:encoding}) of linear term rewrite systems into \pbpostrong. We also show that the encoding is an embedding (Theorem~\ref{thm:encoding:is:embedding}). In the next section, we prove that the embedding is globally termination-preserving.
 
\newcommand{\upcxtnode}{\mathcal{C}}

\newcommand{\graphof}[1]{\mathrm{graph}(#1)}
\newcommand{\rootof}[1]{\mathrm{root}(#1)}

For defining the encoding of terms as graphs, the auxiliary notion of a rooted graph is convenient.

\newcommand{\groot}[1]{\tikz[baseline=(n.base)] \node [circle,draw=cblue!50!black,dotted,fill=cblue!10,inner sep=0.5mm] (n) {$#1$};}

\begin{definition}[Rooted Graph]
    A \emph{rooted graph} $(G, r)$ consists of a graph $G$ and a distinguished \emph{root} $r \in V_G$. We let $\graphof{(G,r)} = G$ and $\rootof{(G,r)} = r$.
    
    We usually omit $\graphof{\ldots}$  in places where a non-rooted graph is expected, since  confusion is unlikely to occur. In visual depictions, the root $r$ is highlighted in a circle \groot{r}. 
\end{definition}

\newcommand{\encode}[2]{\mathcal{E}(#1, #2)}

\newcommand{\enc}{\circ}

\begin{definition}[Term Encoding]
    Define the flat lattice $\Sigma^\enc$ for signatures $\Sigma$ by $\Sigma^\enc = (\Sigma \uplus {\nat^{+}})^{\bot,\top}$.
    
    For linear terms $\term{t} \in \ter$, we define the \emph{term encoding $\term{t}^\enc$ of $\term{t}$} as the $\Sigma^\enc$-labeled rooted graph
    $\term{t}^\enc = \mathcal{E}(\term{t},\epsilon)$, where $\encode{\term{t}}{p}$ is defined by clauses
    \[
    \encode{f(t_1, \ldots, t_n)}{p} =
    \begin{tikzcd}[column sep=6mm, row sep=-1mm]
                     & \groot{p^f} \arrow[ld, "1" description] \arrow[rd, "n" description] &                  \\
    \encode{\term{t_1}}{p1} & \cdots                                                      & \encode{\term{t_n}}{pn}
    \end{tikzcd}
    \]
    and $
    \encode{x}{p} = (x^\bot, x)
    $
    for $f \in \Sigma$, $\term{t_1},\ldots,\term{t_n} \in \ter$, $x \in \xvar$ and $p \in \nat^*$.
    The target of an edge pointing towards a rooted graph $(G,p')$ is $p'$.
    In these graphs, the identity of an edge with source $p$ and target $p'$ is $(p, p')$.
\end{definition}

Note that the term encoding always results in a tree, because the terms it operates on are linear.

\begin{definition}[Positions in Term Encodings]
    Analogous to positions in terms $\term{t}$ (Definition~\ref{def:position}), we assign positions to the nodes of $\term{t}^\enc$:
    $\rootof{\term{t}^\enc}$ is assigned position $\epsilon$; and
    if
    \begin{tikzcd}[column sep=0.5cm]
        v \arrow[r, "i"] & w
    \end{tikzcd}
    is an edge in $E_{\term{t}^\enc}$ (for $i \ge 1$) and $v$ has position $p$, then $w$ is assigned position $pi$.
    
    A translated rule $\rho^\enc$ is said to be \emph{applied at position $p$ in $\term{t}^\enc$} if the match morphism $m : L \mono \term{t}^\enc$ maps the root of $L$ onto the vertex with position $p$ in~$\term{t}^\enc$, and establishes a match.
\end{definition}

\newcommand{\uppercontext}[1]{\upcxtnode\mathrm{[}#1\mathrm{]}}
\newcommand{\lowercontext}[1]{{#1}{\downarrow_\xvar}}
\newcommand{\contextclosure}[1]{\uppercontext{\lowercontext{#1}}}

The following definition is used in the setting of rule encodings.

\begin{definition}[Context Closures]
    Let $G_r = (G,r)$ be a rooted graph.
    
    Assume $\upcxtnode \notin V_G$. The \emph{upper context closure} of $G_r$, denoted $\uppercontext{G_r}$, is the $r$-rooted graph obtained by adding a $\top$-labeled vertex $\upcxtnode$ and two $\top$-labeled edges with identities
    $(\upcxtnode, r)$ and $(\upcxtnode, \upcxtnode)$
    to $G$. Sources and targets are given by the first and second projections, respectively.
    
    For $x \in \xvar$, let $x'$ be fresh for $V_G$.
    The \emph{lower context closure} of $G$ w.r.t.\ a subset $\xvar \subseteq V_G$, denoted $\lowercontext{G}$, is the $r$-rooted graph obtained as follows: for every $x \in V_G \cap \xvar$, (i)~relabel $x$ to $\top$, and (ii)~add a $\top$-labeled vertex $x'$ and two $\top$-labeled edges $(x,x')$ and  $(x',x')$ to $G$.
    
    The \emph{context closure} of $G_r$ is defined as $\contextclosure{G_r}$.
\end{definition}

\noindent
\begin{minipage}[t]{.45\textwidth}
    \vspace*{-1.5ex}
    \begin{example}
        The term encoding $\term{t}^\enc$ of $\term{t} = f(g(x), a, h(y))$ and its context closure $\contextclosure{\term{t}^\enc}$ are shown on the right.
        Both graphs are rooted in $\epsilon$. (The edge identities are left implicit.)
    \end{example}
\end{minipage}\hfill%
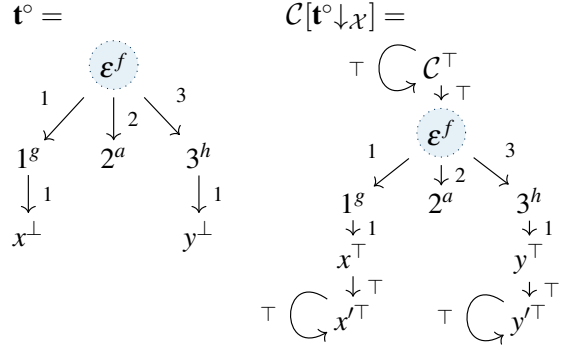
\begin{minipage}[t]{.20\textwidth}
    $\term{t}^\enc =$\\[1ex]
    $
    \begin{tikzcd}[nodes={rectangle,outer sep=1mm,inner sep=0},row sep=7mm,column sep=6mm]
                       & \groot{\epsilon^f} \arrow[ld, "1"'] \arrow[d, "2"] \arrow[rd, "3"] &                    \\
    1^g \arrow[d, "1"] & 2^a                                                        & 3^h \arrow[d, "1"] \\
    x^\bot             &                                                            & y^\bot            
    \end{tikzcd}
    $
\end{minipage}%
\begin{minipage}[t]{.3\textwidth}
    \;\;\;\;$\contextclosure{\term{t}^\enc} =$\\[-1ex]
    $
    \begin{tikzcd}[nodes={rectangle,outer sep=1mm,inner sep=0},row sep=4mm,column sep=6mm]
                                                                 & \upcxtnode^\top \arrow[d, "\top"] \arrow["\top"', loop, distance=2em, in=215, out=150] &                                                              \\
                                                                 & \groot{\epsilon^f} \arrow[ld, "1"'] \arrow[d, "2"] \arrow[rd, "3"]                    &                                                              \\
    1^g \arrow[d, "1"]                                           & 2^a                                                                           & 3^h \arrow[d, "1"]                                           \\
    x^\top \arrow[d, "\top"]                                     &                                                                               & y^\top \arrow[d, "\top"]                                     \\
    x'^\top \arrow["\top"', loop, distance=2em, in=210, out=145] &                                                                               & y'^\top \arrow["\top"', loop, distance=2em, in=210, out=145]
    \end{tikzcd}
    $
\end{minipage}

\begin{definition}[Variable Heads and Symbol Vertices]
\label{def:variable:heads:symbol:vertices}
   For term encodings $\term{t}^\enc$,  the vertices in $x \in V_{\term{t}^\enc} \cap \xvar$ with $\lbl(x) = \bot$ are called \emph{variable heads}, and the remaining vertices labeled from $\Sigma$ are called \emph{symbol vertices}.
\end{definition}

\newcommand{\interface}[1]{\mathcal{I}(#1)}

\begin{definition}[Interface Graph]\label{def:interface:graph}
    The \emph{interface graph} $\interface{\term{t}}$ for a term $\term{t}$~is~the rooted graph $(G', \epsilon)$, where $G'$ is the discrete graph induced by $V_{G'} = \var{\term{t}} \cup \{ \epsilon \}$ and $\lblv_{G'}(v) = \bot$ for all $v \in V_{G'}$.
\end{definition}

\newcommand{\restrict}[2]{#1|_{#2}}

\begin{definition}[Rule Encoding]\label{def:translate:trs:rule}
    The \emph{rule encoding} $\rho^\enc$ of a linear term rewrite rule $\rho : \term{l} \to \term{r}$ over $\Sigma$ into a (linear) \pbpostrong{} rewrite rule over $\Sigma^\enc$-labeled graphs is defined as follows:
    \begin{align*}
        L &= \term{l}^\enc &
        K &= \interface{\term{r}} &
        R &= \term{r}^\enc &
        \\
        L' &= \contextclosure{\term{l}^\enc} & 
        K' &= \contextclosure{\interface{\term{r}}} &
    \end{align*}
    Here we implicitly consider the rooted graphs as graphs by forgetting their roots.
    Each of the morphisms $l$, $r$, $l'$, $t_L$, and $t_K$ map roots to roots and behave as inclusions otherwise.
\end{definition}

Observe that the rule encoding accounts for the special case where the right-hand side $\term{r}$ of the TRS rule is a variable $x$, in which case $r : K \to R$ is the morphism determined by $r(\epsilon) = r(x) = x$. (The case where the left-hand side $\term{l}$ is a variable is excluded by definition.)

\store{proposition}{}{prop:translation:well:defined:and:monicity}{
    In Definition~\ref{def:translate:trs:rule}, all of the morphisms are well-defined and uniquely determined, and the pullback property is satisified. Moreover, morphisms $l$, $l'$, $t_L$ and $t_K$ are monic, and $r$ is monic iff $\term{r}$ is not a variable. \qed
}

\begin{example}[Rule Encoding]
    The TRS rule $\rho = f(x,g(b),y) \to h(g(y), a)$ is encoded as the \pbpostrong{} rewrite rule $\rho^\enc$ given by
    {

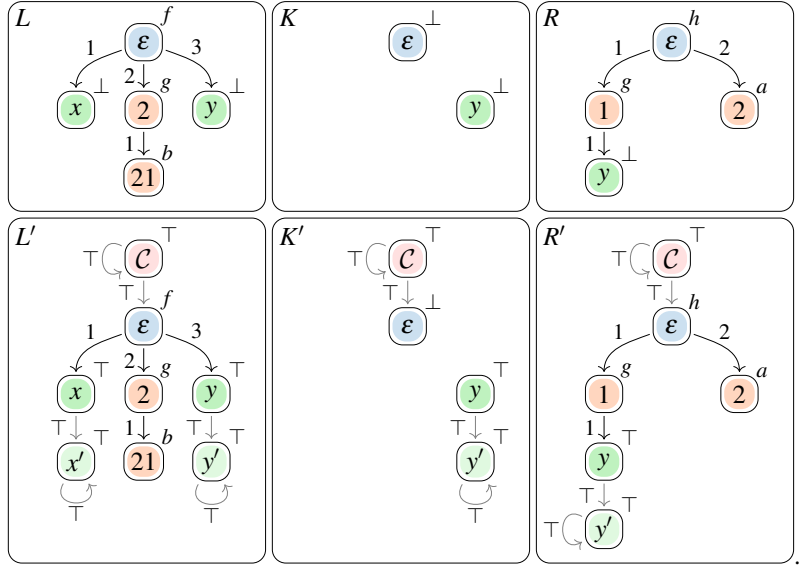
\newcommand{\nodeeps}{\vertex{\epsilon}{cblue!20}}
\newcommand{\nodeOne}{\vertex{1}{corange!25}}
\newcommand{\nodeTwo}{\vertex{2}{corange!25}}
\newcommand{\nodeTwoOne}{\vertex{21}{corange!25}}
\newcommand{\nodex}{\vertex{x}{cgreen!25}}
\newcommand{\nodey}{\vertex{y}{cgreen!25}}
\newcommand{\nodexp}{\vertex{x'}{cgreen!12}}
\newcommand{\nodeyp}{\vertex{y'}{cgreen!12}}
\newcommand{\nodeC}{\vertex{\upcxtnode}{cred!12}}

\begin{center}
  \scalebox{\rulescale}{
  \begin{tikzpicture}[->,node distance=12mm,n/.style={}]
    \graphbox{$L$}{0mm}{0mm}{38mm}{31mm}{1mm}{-6mm}{
      \node [npattern] (eps) [short] {\nodeeps}; \annotate{eps}{$f$};
      \node [npattern] (g) [below of=eps,short] {\nodeTwo};
      \annotate{g}{$g$}
      \node [npattern] (b) [below of=g,short] {\nodeTwoOne};
      \annotate{b}{$b$}
      \node [npattern] (x) [left of=g,short] {\nodex}; \annotate{x}{$\bot$};
      \node [npattern] (y) [right of=g,short] {\nodey}; \annotate{y}{$\bot$};
      
      \draw [epattern] (eps) to[out=-160,in=90] node [above,label,inner sep=1mm] {$1$} (x);
      \draw [epattern] (eps) to node [left,label, inner sep=1mm] {$2$} (g);
      \draw [epattern] (eps) to[out=-20,in=90] node [above,label,inner sep=1mm] {$3$} (y);
      \draw [epattern] (g) to node [left,label,inner sep=1mm] {$1$} (b);
    }
    \graphbox{$K$}{39mm}{0mm}{38mm}{31mm}{1mm}{-6mm}{
      \node [npattern] (eps) [short] {\nodeeps}; \annotate{eps}{$\bot$};
      \node (g) [below of=eps,short] {};
      \annotate{g}{}
      \node [npattern] (y) [right of=g,short] {\nodey}; \annotate{y}{$\bot$};
    }
    \graphbox{$R$}{78mm}{0mm}{38mm}{31mm}{1mm}{-6mm}{
      \node [npattern] (eps) [short] {\nodeeps}; \annotate{eps}{$h$};
      \node (xx) [below of=eps,short] {};
      \node [npattern] (g) [left of=xx,short] {\nodeOne}; \annotate{g}{$g$};
      \node [npattern] (y) [below of=g,short] {\nodey}; \annotate{y}{$\bot$};
      \node [npattern] (a) [right of=xx,short] {\nodeTwo}; \annotate{a}{$a$};
      
      \draw [epattern] (eps) to[out=-160,in=90] node [above,label,inner sep=1mm] {$1$} (g);
      \draw [epattern] (eps) to[out=-20,in=90] node [above,label,inner sep=1mm] {$2$} (a);
      \draw [epattern] (g) to node [left,label,inner sep=1mm] {$1$} (y);
    }
    \graphbox{$L'$}{0mm}{-32mm}{38mm}{51mm}{1mm}{-16mm}{
      \node [npattern] (eps) [short] {\nodeeps}; \annotate{eps}{$f$};
      \node [npattern] (g) [below of=eps,short] {\nodeTwo};
      \annotate{g}{$g$}
      \node [npattern] (b) [below of=g,short] {\nodeTwoOne};
      \annotate{b}{$b$}
      \node [npattern] (x) [left of=g,short] {\nodex}; \annotate{x}{$\top$};
      \node [npattern] (y) [right of=g,short] {\nodey}; \annotate{y}{$\top$};
      \node [npattern] (xp) [below of=x,short] {\nodexp}; \annotate{xp}{$\top$};
      \node [npattern] (yp) [below of=y,short] {\nodeyp}; \annotate{yp}{$\top$};
      \node [npattern] (c) [above of=eps,short] {\nodeC}; \annotate{c}{$\top$};

      \draw [epattern] (eps) to[out=-160,in=90] node [above,label,inner sep=1mm] {$1$} (x);
      \draw [epattern] (eps) to node [left,label, inner sep=1mm] {$2$} (g);
      \draw [epattern] (eps) to[out=-20,in=90] node [above,label,inner sep=1mm] {$3$} (y);
      \draw [epattern] (g) to node [left,label,inner sep=1mm] {$1$} (b);
       \draw [eset] (x) to node [left,label, inner sep=1mm] {$\top$} (xp);
       \draw [eset] (y) to node [left,label, inner sep=1mm] {$\top$} (yp);
       \draw [eset,loop=270,looseness=3] (xp) to node [below,label] {$\top$} (xp);
       \draw [eset,loop=270,looseness=3] (yp) to node [below,label] {$\top$} (yp);
       \draw [eset] (c) to node [left,label, inner sep=1mm] {$\top$} (eps);
       \draw [eset,loop=180,looseness=3] (c) to node [left,label] {$\top$} (c);
    }
    \graphbox{$K'$}{39mm}{-32mm}{38mm}{51mm}{1mm}{-16mm}{
      \node [npattern] (eps) [short] {\nodeeps}; \annotate{eps}{$\bot$};
      \node (g) [below of=eps,short] {};
      \node [npattern] (y) [right of=g,short] {\nodey}; \annotate{y}{$\top$};
      \node [npattern] (yp) [below of=y,short] {\nodeyp}; \annotate{yp}{$\top$};
      \node [npattern] (c) [above of=eps,short] {\nodeC}; \annotate{c}{$\top$};
      
       \draw [eset] (y) to node [left,label, inner sep=1mm] {$\top$} (yp);;
       \draw [eset,loop=270,looseness=3] (yp) to node [below,label] {$\top$} (yp);
       \draw [eset] (c) to node [left,label, inner sep=1mm] {$\top$} (eps);
       \draw [eset,loop=180,looseness=3] (c) to node [left,label] {$\top$} (c);
    }
    \graphbox{$R'$}{78mm}{-32mm}{38mm}{51mm}{1mm}{-16mm}{
      \node [npattern] (eps) [short] {\nodeeps}; \annotate{eps}{$h$};
      \node (xx) [below of=eps,short] {};
      \node [npattern] (g) [left of=xx,short] {\nodeOne}; \annotate{g}{$g$};
      \node [npattern] (y) [below of=g,short] {\nodey}; \annotate{y}{$\top$};
      \node [npattern] (a) [right of=xx,short] {\nodeTwo}; \annotate{a}{$a$};
      
      \draw [epattern] (eps) to[out=-160,in=90] node [above,label,inner sep=1mm] {$1$} (g);
      \draw [epattern] (eps) to[out=-20,in=90] node [above,label,inner sep=1mm] {$2$} (a);
      \draw [epattern] (g) to node [left,label,inner sep=1mm] {$1$} (y);
      
      \node [npattern] (yp) [below of=y,short] {\nodeyp}; \annotate{yp}{$\top$};
      \node [npattern] (c) [above of=eps,short] {\nodeC}; \annotate{c}{$\top$};
      
       \draw [eset] (y) to node [left,label, inner sep=1mm] {$\top$} (yp);;
       \draw [eset,loop=-180,looseness=3] (yp) to node [left,label] {$\top$} (yp);
       \draw [eset] (c) to node [left,label, inner sep=1mm] {$\top$} (eps);
       \draw [eset,loop=180,looseness=3] (c) to node [left,label] {$\top$} (c);
    }
  \end{tikzpicture}.
  }
\end{center}

}

    An application of this rule can be thought of as binding head variable $x$ and $y$ of $L$ to the roots of two subterms. These subterms and the context are then uniquely captured by $L'$ (by virtue of the strong match property), and correctly rearranged around $R$ by the rewrite step.
\end{example}

Rule encodings extend to rewrite system encodings in the obvious way.

\begin{definition}[Rewrite System Encoding]
The \emph{rewrite system encoding} $R^\enc$ of a linear TRS $R$ is $\{ \rho^\enc \mid \rho \in R \}$.
\end{definition}

All the encodings we have introduced have obvious inverses.

\newcommand{\decode}{\mathit{decode}}
\begin{definition}[Decoding]
For term/rule/system encodings $x^\enc$, we define the inverse $\decode(x^\enc) = x$.
\end{definition}

\begin{proposition}[Root Mapping Determines Adherence]
\label{prop:root:determines:adherence}
Let $\rho = \term{l} \to \term{r}$ be a linear term rewrite rule. If $\rho^\enc$ is applied at position $p$ in $\term{s}^\enc$, then a unique $\alpha : \term{s}^\enc \to \contextclosure{\term{l}^\enc}$ exists that establishes a strong match, i.e., that makes
\[
   \begin{tikzpicture}[node distance=16mm, l/.style={inner sep=1mm},baseline=-7.5mm,v/.style={node distance=15mm}]
        \node (G) {$\term{s}^\enc$};
        \node (L) [left of=G] {$\term{l}^\enc$};
          \draw [>->] (L) to node [above,l] {$m$} (G);
        \node (L2) [v,below of=L] {$\term{l}^\enc$};
          \draw [>->] (L) to node [left,l] {$1_{\term{l}^\enc}$} (L2);
          \node at ([shift={(4mm,-4mm)}]L) {PB};
        \node (LP) [v,below of=G] {$\contextclosure{\term{l}^\enc}$};
          \draw [>->] (L2) to node [above,l] {$t_L$} (LP);
          \draw [->] (G) to node [left,l] {$\alpha$} (LP);
          \end{tikzpicture}
    \]
a pullback square.
\end{proposition}
\begin{proof}
By definition of applying at a position $p$, $m$ maps $\rootof{\term{l}^\enc}$ onto position $p$ of $\term{s}^\enc$, fully determining $m$ to map nodes with position $q$ in $\term{l}^\enc$ onto nodes with position $pq$ in $\term{s}^\enc$. A node in $\term{l}^\enc$ is either a symbol vertex or a variable head. For symbol vertices, any $m$ must preserve labels.
Variable heads (labeled with $\bot$) are mapped by $m$ onto either (i)~vertex $\sigma(x)$ labeled with $\bot$ if $\sigma(x) \in \xvar$ is a variable, or (ii)~vertex $pq$ labeled with $f \in \Sigma$ if $x$ is substituted for some non-variable term $\sigma(x) = f(\term{t_1}, \ldots, \term{t_n})$ ($n \geq 0$).
    
On the image $m(\term{l}^\enc)$, define $\alpha$ such that $t_L = \alpha \circ m$. The labels of symbol vertices are thereby preserved, and the labels in the head variables of $\term{l}^\enc$ are increased to $\top$. The elements not in $m(\term{l}^\enc)$ can be mapped onto the appropriate elements added by the context closure, and only in one way as to not overlap with $t_L$. Because $t_L$ does not map onto these closures, pulling $\alpha$ along $t_L$ gives the required pullback square.
\end{proof}

\store{lemma}{Match Determinism}{lem:translation:determinism}{
Let $\rho = \term{l} \to \term{r}$ be a linear term rewrite rule. 
    If $\rho^\enc$ is applied at position $p$ in $\term{s}^\enc$ and gives rise to a step $\term{s}^\enc \to G$, then $G$ is uniquely determined up to isomorphism.
}
\begin{proof}
By Proposition~\ref{prop:root:determines:adherence}, adherence $\alpha$ is completely determined, and by general categorial properties, the pullback of $\alpha$ along $l'$ gives a unique result up to isomorphism, and so does the final pushout.
\end{proof}

\begin{proposition}
\label{prop:mono:encoding:allows:decomposition}
    If $m : \term{l}^\enc \mono \term{s}^\enc$ is a mono, then $\term{s}^\enc = (C[\term{l}\sigma])^\enc$ for some context $C$ and substitution $\sigma$. Moreover, the position of $m(\rootof{\term{l}^\enc})$ in $\term{s}^\enc$ equals the position of $\hole$ in $C[\;]$.
\end{proposition}
\begin{proof}
    By monicity of $m$, the tree structure of $\term{l}^\enc$ is preserved into $\term{s}^\enc$. The labels of symbol vertices and edges are also preserved, since $\term{s}^\enc$ has no occurrences of $\top$. This also means that, for every vertex $v$ of $\term{l}^\enc$, $v$ and $m(v)$ have the same number of outgoing edges, since encodings preserve arities.

    A variable head $x \in V_{\term{l}^\enc}$ is mapped onto a vertex $m(x)$, which is either a variable head with label $\bot$, or a symbol vertex labeled with some $f \in \Sigma$ and a subtree underneath. 
    
    Let $p$ be the position of $m(\rootof{\term{l}^\enc})$ in $\term{s}^\enc$.
    Define $C$ as the context obtained from $\term{s}$ by replacing the subterm at position $p$ by $\hole$.
    Define the substitution $\sigma$, for every $x \in \var{\term{l}}$, by $\sigma(x) = \term{s}|_{pq_x}$ where $q_x$ is the position of $x$ in $\term{l}$.
    Then the claim follows since $m$ maps $x$ in $\term{l}^\enc$ to the position $pq_x$ in $\term{s}^\enc$, and the subtree rooted at this position is $(\term{s}|_{pq_x})^\enc$.
\end{proof}

\begin{lemma}[$(\cdot)^\enc$ Is Step-Preserving]
\label{lem:st:ssts}
Let $\rho = \term{l} \to \term{r}$ be a linear term rewrite rule. If $\term{s} \to \term{t}$ via $\rho$ at position $p$, then $\term{s}^\enc \to \term{t}^\enc$ via $\rho^\enc$ at position $p$.
\end{lemma}
\begin{proof}
    By the definition of a term rewrite step,
    $\term{s} = C[\term{l}\sigma]$ and $\term{t} = C[\term{r}\sigma]$ for some context $C$ and substitution $\sigma$, and $\term{l}\sigma$ is at position $p$ in $C[\term{l}\sigma]$.
    
    By the definitions of encodings and a \pbpostrong rewrite step, we must show that the diagram
    \begin{equation}
    \label{eq:diagram:step:preserving}
      \begin{tikzpicture}[node distance=20mm, l/.style={inner sep=1mm},baseline=-7.5mm,v/.style={node distance=13mm}]
        \node (G) {$(C[\term{l}\sigma])^\enc$};
        \node (L) [left of=G] {$\term{l}^\enc$};
          \draw [>->] (L) to node [above,l] {$m$} (G);
        \node (L2) [v,below of=L] {$\term{l}^\enc$};
          \draw [>->] (L) to node [left,l] {$1_{\term{l}^\enc}$} (L2);
          \node at ([shift={(4mm,-4mm)}]L) {PB};
        \node (LP) [v,below of=G] {$\contextclosure{\term{l}^\enc}$};
          \draw [>->] (L2) to node [above,l] {$t_L$} (LP);
          \draw [->] (G) to node [left,l] {$\alpha$} (LP);
         \node (GKP) [right= of G] {$G_K$};
         \draw [>->] (GKP) to node [above,l] {$g_L$} (G);
         \node (KP) [v,below of=GKP] {$\contextclosure{\interface{\term{r}}}$};
         \draw [->] (GKP) to node [right,l] {$u'$} (KP);
         \draw [>->] (KP) to node [above,l] {$l'$} (LP);
         \node at ([shift={(-4mm,-4mm)}]GKP) {PB};
         \node (K) [v,above of=GKP] {$\interface{\term{r}}$};
         \draw [>->, dotted] (K) to node [right,l] {$!u$} (GKP);
         \node (R) [right of=K] {$\term{r}^\enc$};
         \draw [->] (K) to node [above,l] {$r$} (R);
         \node (GP) [v,below of=R] {$(C[\term{r}\sigma])^\enc$};
         \draw [->] (GKP) to node [below,l] {$g_R$} (GP);
         \draw [>->] (R) to node [right,l] {$w$} (GP);
         \node at ([shift={(-4mm,4mm)}]GP) {PO};
      \end{tikzpicture}
    \end{equation}
    holds for some $G_K$ and the various morphisms that are not fixed by $\rho^\enc$ (including $\alpha$), and where $m$ maps $\rootof{\term{l}^\enc}$ onto position $p$ of $(C[\term{l}\sigma])^\enc$. Note that $g_L$ is a mono by Proposition~\ref{prop:translation:well:defined:and:monicity} and stability of monos under pullbacks, and $w$ is a mono by Proposition~\ref{prop:mono:stable:pushout}.
    
    By Proposition~\ref{prop:root:determines:adherence}, $m$ and $\alpha$ exist and they exist uniquely. It is then straightforward to check that the middle pullback extracts the subgraphs corresponding to the context $C$ and to every subterm bound to a variable $x \in \var{\term{l}} \cap \var{\term{r}}$, and that the pushout performs the appropriate gluing around pattern $\term{r}^\enc$, with $(C[\term{r}\sigma])^\enc$ as the result.
\end{proof}

\begin{lemma}[$(\cdot)^\enc$ Is Closed]
\label{lem:encoding:is:closed}
Let $\rho = \term{l} \to \term{r}$ be a linear term rewrite rule. If $\term{s}^\enc \to G$ via $\rho^\enc$ then $G \iso \term{t}^\enc$ for some term $\term{t}$ with $\term{s} \to \term{t}$.
\end{lemma}
\begin{proof}
    Assume $\term{s}^\enc \to G$ via $\rho^\enc $ at position $p$.
    Then by Proposition~\ref{prop:mono:encoding:allows:decomposition} we have $\term{s}^\enc = (C[\term{l}\sigma])^\enc$ for some context $C$ and substitution $\sigma$ such that $\term{s}(p) = \hole$.
    Then $\term{s} = C[\term{l}\sigma] \to C[\term{r}\sigma] = \term{t}$ via $\rho$ at position $p$.
    Thus $\term{s}^\enc \to \term{t}^\enc$ via $\rho^\enc$ at position $p$ by Lemma~\ref{lem:st:ssts}.
    Then we have $G \iso \term{t}^\enc$ by Lemma~\ref{lem:translation:determinism}.
\end{proof}

\begin{theorem}
\label{thm:encoding:is:embedding}
The encoding $(\cdot)^\enc$ is an embedding.
\end{theorem}
\begin{proof}
    From Lemma~\ref{lem:st:ssts} and Lemma~\ref{lem:encoding:is:closed}.
\end{proof}

\section{The Embedding Preserves Termination Globally}
\label{sec:termination}

From the fact that the encoding is step-preserving (Lemma~\ref{lem:st:ssts}), the following is almost immediate.

\store{lemma}{}{lemma:terminating:on:graphs:terminating:on:terms}{
Let $R$ be a linear TRS.
    If $R^\enc$ is terminating on {\FinGraphLattice{\Sigma^\enc}}, then $R$ is terminating. \qed
}

It is obvious that the other direction holds if the category \FinGraphLattice{\Sigma^\enc} is restricted to graphs that are term encodings; so we have \emph{local termination}~\cite{termination:local:2009,termination:local:2010,termination:automata:2015}. However, in this subsection we will show that the direction holds \emph{globally}. Thus, in particular, the finite graphs may be disconnected, cyclic, and labeled arbitrarily from $\Sigma^\enc$.

Our overall proof strategy is as follows. First, we show that it suffices to restrict to cycle-free graphs $G$ (Corollary~\ref{corollary:terminating:iff:cycle-free:terminating}). Then, we show that an infinite rewrite sequence on cycle-free $G$ contains (in some sense) an infinite rewrite sequence on term encodings, and therefore on terms~(Theorem~\ref{thm:terminating:terms:iff:terminating:graphs}).

\newcommand{\dropcycles}[1]{[#1]}

\begin{definition}[Undirected Path]
  Let $n \in \nat$. An \emph{undirected path} of length~$n$ from node $v_1$ to $v_{n+1}$ in a graph $G$ is a sequence $v_1 \, e_1 \, v_2 \, v_2 \cdots v_n \, e_n \, v_{n+1}$ where
  $v_1,v_2,\ldots,v_{n+1}$ are nodes of $G$ and
  $e_1,e_2,e_3,\ldots,e_n$ are edges of $G$ such that
  \(
  (v_i,v_{i+1}) \in \{\, (s(e_i),t(e_i)),\; (t(e_i),s(e_i)) \,\}  \text{ for every $1 \le i \le n$.}
  \)
  
  The path is an \emph{undirected cycle} if moreover $n > 0$, $v_1 = v_{n+1}$ and $e_i \ne e_j$ for all $0 < i < j \leq n$.
  A \emph{cycle edge} (\emph{cycle node}) is an edge (node) that is part of an undirected cycle. A graph is \emph{cycle-free} if it does not contain undirected cycles.
\end{definition}

\begin{example}
  A path of length 1 is an undirected cycle iff its only edge $e$ is a loop, that is, $s(e) = t(e)$.
  Two edges between two nodes always constitute an undirected cycle of length 2 (irrespective of the direction of the edges).
\end{example}

\store{proposition}{}{lem:cycle:two:paths}{
    Edge $e$ is a cycle edge iff there exists an undirected path from $s(e)$ to $t(e)$ that does not include $e$.
}
\begin{proof}
    If $s(e) = t(e)$, one path is the empty path. Obvious otherwise.
\end{proof}

\begin{proposition}\label{prop:cycle:edge:mono}
    If $e$ is a cycle edge in $G$ and $\phi: G \mono H$ a mono, then $\phi(e)$ is a cycle edge in $H$. \qed
\end{proposition}

Although monos preserve the cycle edge property, morphisms do not generally do so (consider a morphism that identifies two parallel edges). However, for adherence morphisms $\alpha$ we have the following result.

\store{lemma}{}{lem:cycle:edges}{
  Consider the \pbpostrong{} match square (the leftmost square of the rewrite step diagram) with a host graph $G_L$.
  Suppose that $e$ is a cycle edge in $G_L$ and $\alpha(e) = t_L(e')$ for some $e' \in E_L$. Then $\alpha(e)$ is a cycle edge in $L'$.
}
\begin{proof}
    Let $\sigma_1$ be the path just consisting of $e$.
    By Proposition~\ref{lem:cycle:two:paths} there also exists an undirected path from $s(e)$ to $t(e)$ in $G_L$ that does not include $e$.
    Since premorphisms preserve undirected paths, $\alpha(\sigma_1)$ and $\alpha(\sigma_2)$ are undirected paths from $\alpha(s(e))$ to $\alpha(t(e))$ in $L'$. If $\alpha(e)$ is not a cycle edge, then paths $\alpha(\sigma_1)$ and $\alpha(\sigma_2)$ both include $\alpha(e)$ by Proposition~\ref{lem:cycle:two:paths}. Thus $\alpha$ maps two distinct edges in $G_L$ onto $\alpha(e) = t_L(e')$. Since $L$ is the $\alpha$-preimage of $t_L(L)$, $\alpha \circ m = t_L$ also maps two distinct edges onto $t_L(e)$. This contradicts that $t_L$ is monic. So $\alpha(e)$ is a cycle edge.
\end{proof}

\store{lemma}{Cycle-Preserving Pullback}{lem:cycle:preserving:pullback}{
    If for $\tau = G \xrightarrow{g} X \xleftarrow{h} H$, (i)~$\sigma$ is an undirected cycle in $G$, (ii)~$g(\sigma)$ lies in the image of $h$, and (iii)~the pullback for $\tau$ is $G \xleftarrow{g'} Y \xrightarrow{h'} H$, then every edge $e \in g'^{-1}(\sigma)$ is a cycle edge in $Y$. \qed
}

\begin{definition}[Cycle Edge Removal]
  For a graph $G$, we let $\dropcycles{G}$ denote the graph obtained by deleting all cycle edges from $G$.
\end{definition}

\store{lemma}{}{lem:drop:cycles}{
    Let $\rho: \term{l} \to \term{r}$ be a linear term rewrite rule over $\Sigma$.
    If there is a rewrite step $G_L \stackrel{\rho^\enc}{\to} G_R$
    on graphs over~$\Sigma^\enc$,
    then also $\dropcycles{G_L} \stackrel{\rho^\enc}{\to} \dropcycles{G_R}$.
}

\begin{proof}
    By the definition of a rewrite step and substituting for the translation of $\rho$, we have the following arrangement of objects and morphisms
    \begin{center}
      \begin{tikzpicture}[node distance=24mm and 16mm, l/.style={inner sep=1mm},baseline=-7.5mm,v/.style={node distance=13mm}]
        \node (G) {$G_L$};
        \node (L) [left of=G] {$\term{l}^\enc$};
          \draw [>->] (L) to node [above,l] {$m$} (G);
        \node (L2) [v,below of=L] {$\term{l}^\enc$};
          \draw [>->] (L) to node [left,l] {$1_{\term{l}^\enc}$} (L2);
          \node at ([shift={(4mm,-3mm)}]L) {PB};
        \node (LP) [v,below of=G] {$\contextclosure{\term{l}^\enc}$};
          \draw [>->] (L2) to node [above,l] {$t_L$} (LP);
          \draw [->] (G) to node [left,l] {$\alpha$} (LP);
         \node (GKP) [right of=G] {$G_K$};
         \draw [>->] (GKP) to node [above,l] {$g_L$} (G);
         \node (KP) [v,below of=GKP] {$\contextclosure{\interface{\term{r}}}$};
         \draw [->] (GKP) to node [right,l] {$u'$} (KP);
         \draw [>->] (KP) to node [above,l] {$l'$} (LP);
         \node at ([shift={(-4mm,-3mm)}]GKP) {PB};
         \node (K) [v,above of=GKP] {$\interface{\term{r}}$};
         \draw [>->, dotted] (K) to node [right,l] {$!u$} (GKP);
         \node (R) [right of=K] {$\term{r}^\enc$};
         \draw [->] (K) to node [above,l] {$r$} (R);
         \node (GP) [v,below of=R] {$G_R$};
         \draw [->] (GKP) to node [below,l] {$g_R$} (GP);
         \draw [>->] (R) to node [right,l] {$w$} (GP);
         \node at ([shift={(-4mm,3mm)}]GP) {PO};
      \end{tikzpicture}
    \end{center}
    for some $G_K$. Many of the morphisms are fixed by the rule $\rho^\enc$. Note that $g_L$ is a mono by Proposition~\ref{prop:translation:well:defined:and:monicity} and stability of monos under pullbacks.
    
    Observe that $t_L(\term{l}^\enc)$ does not contain cycle edges~(Definition~\ref{def:translate:trs:rule}). Hence by Lemma~\ref{lem:cycle:edges}, $\alpha$ must map every cycle edge of $G_L$ into one of the edges created by constructing the context closure $\contextclosure{\term{l}^\enc}$ of $\term{l}^\enc$.

    Now suppose that we replace $G_L$ by $\dropcycles{G_L}$ in the diagram. Then the middle pullback object $G'_K$ is obtained by removing from $G_K$ the set of edges $C \subseteq E_{G_K}$ that mono $g_L$ maps into a cycle edge of $G_L$. Since monos preserve cycle edges, every cycle edge of $G_K$ is in $C$. Moreover, using Lemma~\ref{lem:cycle:preserving:pullback}, $C$ contains only cycle edges. Hence $G'_K = \dropcycles{G_K}$.
    
    Similarly, the pushout object replacement $G'_R$ for $G_R$ is obtained by removing from $G_R$ the set of edges $C \subseteq E_{G_R}$ that have a cycle edge $g_R$-preimage in $G_K$. Since an undirected path $\rho$ in $G_R$ is an undirected cycle iff $\rho$ is in the range of $g_R$ and $g_R^{-1}(\rho)$ is an undirected cycle,
    $G'_R = \dropcycles{G_R}$.
\end{proof}

As a direct consequence of Lemma~\ref{lem:drop:cycles} we obtain the following.

\store{corollary}{}{corollary:terminating:iff:cycle-free:terminating}{
  Let $R$ be a linear TRS over~$\Sigma$. $R^\enc$ admits an infinite rewrite sequence on all graphs iff $R^\enc$ admits an infinite rewrite sequence on cycle-free graphs. \qed
}

Thus, in order to prove that termination of $R$ implies termination of $R^\enc$ in \FinGraphLattice{\Sigma^\enc}, it suffices to restrict attention to finite, cycle-free graphs. However, not all such graphs are term-like: graphs may be arbitrarily labeled from $\Sigma^\enc$, non-rooted and disconnected. So a further argument is needed.

\newcommand{\wi}{\mathrm{I}}
\newcommand{\wo}{\mathrm{O}}
\newcommand{\wf}{\wi \wedge \wo}
\newcommand{\iif}{\neg\wi \wedge \wo}
\newcommand{\oif}{\wi \wedge \neg\wo}
\newcommand{\fif}{\neg\wi \wedge \neg\wo}
\newcommand{\alt}{\mathrm{A}}
\newcommand{\botalt}{\bot\text{-}\alt}

\begin{definition}[Well-Formedness]\label{def:node:classification}
    Let $\Sigma$ be a signature, and $G$ a graph with labels from $\Sigma^\enc$.
    A node $v \in V_G$ with label $l \in \Sigma \cup \{ \bot, \top \} \cup \mathbb{N}^+$ is \emph{in-well-formed}~($\wi$) if it has at most one incoming edge; and it is \emph{out-well-formed}~($\wo$) if $l \in \Sigma$, and $v$ has precisely $\#{l}$ outgoing edges, labeled with $1$, $2$, \ldots, $\#{l}$.
\end{definition}

\begin{definition}[Good and Bad Nodes]
    A node $v \in V_G$ is called \emph{good} if $v$ is $\wo$ and all of $v$'s children in $G$ are $\wi$. Nodes that are not good are \emph{bad}.
\end{definition}

We will use the distinction between good and bad nodes to define a kind of partitioning on graphs $G$, which we call a \emph{zoning}. For cycle-free graphs, each zone will be seen to correspond to a term encoding in a qualified sense. (Some edges of $G$ will not be part of any zone of $G$.) Since most results related to zoning hold not only for non-cycle-free graphs, we will use minimal assumptions where possible (in particular, note that (directed) acyclicity is a weaker condition than cycle-freeness). We do assume finiteness globally.

\newcommand{\partitioning}{zoning\xspace}
\newcommand{\partition}{zone\xspace}
\newcommand{\partitions}{zones\xspace}

\begin{definition}[Zoning]
    A \emph{\partitioning} of $G$ divides up $G$ into \emph{\partitions}, which are subgraphs of $G$. The zoning is iteratively constructed as follows:
    \begin{itemize}
        \item Initially, every node of $G$ forms its own \partition.
        \item At each subsequent iteration, if an edge $e$ is not included in a zone and $s(e)$ is good, join the zones of $s(e)$ and $t(e)$ along $e$. (If $s(e)$ and $t(e)$ are in the same zone $Z$, this is the same as adding $e$ to zone $Z$.)
        \item The algorithm terminates if the previous step can no longer be applied.
    \end{itemize}
\end{definition}

\begin{definition}[Bridge]
    A \emph{bridge} is an edge $e \in E_G$ not included in any zone of $G$.
\end{definition}

\begin{proposition}
    The \partitioning of a graph $G$ is unique, and any zone is a connected subgraph. \qed
\end{proposition}

\begin{proposition}
\label{prop:nodes:are:I}
  If $e \in E_Z$ is included in zone $Z$, then $t(e)$ is $\wi$.
\end{proposition}
\begin{proof}
 Since $e$ was joined along, $s(e)$ is good, and hence $t(e)$ is $\wi$.
\end{proof}

\begin{definition}[Root]
A node $v \in V_Z$ without a parent inside zone $Z$ is called a \emph{root} for $Z$.
\end{definition}

\begin{proposition}
\label{prop:ancestor:in:zone}
Within a zone $Z$, for any two nodes $u,v \in V_Z$, there is a node $x \in V_Z$ such that $u \leftarrow^* x \to^* v$ (using edges included in $Z$).
\end{proposition}
\begin{proof}
Because any zone is connected, there is an undirected path between $u$ and $v$ within $Z$. This path cannot contain a segment of the form $a \to c \leftarrow b$, for then $c$ would not be $\wi$, contradicting Proposition~\ref{prop:nodes:are:I}. Hence the path must be of the form $u \leftarrow^* x \to^* v$ for some $x \in V_Z$.
\end{proof}

\begin{corollary}
If a zone has a root, it is unique.
\end{corollary}

\begin{proposition}
If a zone $Z$ is acyclic, it has a root.
\end{proposition}
\begin{proof}
If not, following the edges in $Z$ backwards would reveal a directed cycle in $Z$.
\end{proof}

\begin{proposition}
\label{prop:zone:directed:tree}
If a zone $Z$ is acyclic, then $Z$ is a directed tree.
\end{proposition}
\begin{proof}
As follows from the preceding propositions, $Z$ is connected and each zone has a unique root $u$.

By Proposition~\ref{prop:ancestor:in:zone}, $u$ has a path to every node $v$ in $Z$. Such an (acyclic) path is moreover unique, for otherwise the first point at which these paths join is not $\wi$, contradicting Proposition~\ref{prop:nodes:are:I}. Thus $Z$ is a directed tree. 
\end{proof}

We also have the following general characterization of bridges.

\begin{proposition}[On Bridges]
\label{prop:on:bridges}
The source of a bridge is a bad leaf of a zone, and the target of a bridge is a root of a zone.
\end{proposition}
\begin{proof}
If $e$ is a bridge, $s(e)$ must be bad. If $s(e)$ is bad, none of its outgoing edges have been joined along. Hence $s(e)$ is a leaf in $Z$.

If a bridge $e$ targets a non-root $t(e)$ of a zone $Z$, then $t(e)$ is not $\wi$, since it has at least two incoming edges. Thus the parent $p$ of $t(e)$ inside $Z$ is bad.
But this contradicts that $p$ must be good since it has an edge to $t(e)$ inside $Z$. Hence $t(e)$ must be a root.
\end{proof}

Although acyclic zones are directed trees, not every zone corresponds directly to a term encoding $\term{t}^\enc$ for some term~$\term{t}$. For instance, for the 3-zone graph 
\begin{tikzcd}[column sep=4mm]
f \arrow[r, "1"] & a & f \arrow[l, "1"']
\end{tikzcd}, with $\#(f) = 1$ and $\#(a) = 0$, only the \partition containing the node labeled with $a$ corresponds to a term encoding. But we have the following result.

\store{proposition}{Zones as Term Encodings}{prop:zone:to:term:encoding}{
    If every bad node of an acyclic zone $Z$ is relabeled with $\bot$, then $Z$ is isomorphic to a term encoding~$\term{t}^\enc$.
}

\begin{proof}
    Every acyclic zone is structurally a directed tree. All inner nodes (and some leaves labeled with constants $a \in \Sigma$) are good, meaning they are labeled with $\Sigma$ and out-well-formed; and all of their children are in-well formed and included into the zone by the zoning algorithm. Since bad nodes are leaves,  relabeling them with $\bot$ essentially makes them represent variables. To establish an isomorphism between a zone and a term encoding, one simply has to rename the identity of every good node to its position in this tree, and the identity of every bad node to some unique $x \in \xvar$.
\end{proof}

We will now show that relabeling bad nodes with $\bot$ does not meaningfully affect the rewriting behavior in a graph $G$. Intuitively, this is because matches cannot cross zones, as shown by the following results. Recall the terminology of Definition~\ref{def:variable:heads:symbol:vertices}.

\begin{lemma}
\label{lemma:symbo:vertices:to:good:nodes}
A match morphism $m : \term{l}^\enc \to G$ (for a rule encoding $\rho^\enc$) maps symbol vertices $v \in V_{\term{l}^\enc}$ onto good nodes.
\end{lemma}
\begin{proof}
We must show that $m(v)$ is $\wo$ and that all of $m(v)$'s children are $\wi$.

First, we show that \emph{$m(v)$ is $\wo$.} Because $v$ is a symbol vertex, $\lbl(v) \in \Sigma$. Since morphisms do not decrease labels, either (a)~$\lbl(v) = \lbl(m(v))$, or (b)~$\lbl(v) < \lbl(m(v))$.

In case (a), we must show that $m(v)$ has precisely $\#{(\lbl(m(v)))} = \#{(\lbl(v))}$ outgoing edges labeled with $1,2,\ldots,\#{(\lbl(v))}$. By monicity of $m$ and the definition of rule encodings, we know that it has these edges at least once. Moreover, $m(v)$ cannot have additional outgoing edges, since these cannot be suitably mapped by $\alpha$ into $L'$ without violating the strong match property. 

In case (b), we obtain a contradiction. For note that $t_L : \term{l^\enc} \to \contextclosure{\term{l^\enc}}$ preserves labels for nodes labeled from $\Sigma$, so that $\lbl(t_L(v)) = \lbl(v)$. Furthermore, since $m$ enables a rewrite step, $t_L = \alpha \circ m$ and hence $\lbl((\alpha \circ m)(v)) = \lbl(v)$. This implies that $\alpha$ decreases the label on $m(v)$, which is not allowed by the $\leq$ requirement on morphisms.

Second, we establish that \emph{all of $m(v)$'s children are $\wi$.} 
Observe that for symbol vertices $v$, all incoming edges of children of $t_L(v)$ (i) have their source in $t_L(v)$ and (ii) are in the image of $t_L$.
For a contradiction, assume a child $u$ of $m(v)$ has multiple incoming edges $e,e'$. Then using that $\alpha(m(v)) = t_L(v)$ (by the strong match property) and that $\alpha(u)$ is a child of $t_L(v)$, by observation~(i) $\alpha(s(e)) = \alpha(s(e'))$. Since there are no parallel edges in $L'$, $\alpha(e) = \alpha(e')$. By~(ii) $\alpha(e)$ is in the image of $t_L$. Thus multiple elements are mapped onto the same element in $L'$. This violates the strong match property. Contradiction.
\end{proof}

\store{lemma}{Matches Respect Boundaries}{lemma:matches:in:one:zone}{
    Let $\rho = \term{l} \to \term{r}$ be a TRS rule, and consider the translation $\rho^\enc$.
    Then for any match morphism $m : \term{l}^\enc \mono G$, the image $m(\term{l^\enc})$ lies in precisely one \partition.
}
\begin{proof}
    Because $\term{l^\enc}$ is connected, so is $m(\term{l}^\enc)$. So if a counterexample to the lemma exists, it involves a bridge. Let $m(e)$, the image of an $e \in E_{\term{l}^\enc}$, be such a bridge. By Proposition~\ref{prop:on:bridges}, $s(m(e)) = m(s(e))$ is a bad leaf of a zone $Z$. Hence $s(e) \in V_{\term{l}^\enc}$ is a variable head by the contrapositive of Proposition~\ref{lemma:symbo:vertices:to:good:nodes}. Since variable heads are leaves in $\term{l}^\enc$, this contradicts that $s(e)$ has $e \in E_{\term{l}^\enc}$ for an outgoing edge.
\end{proof}

\begin{figure}
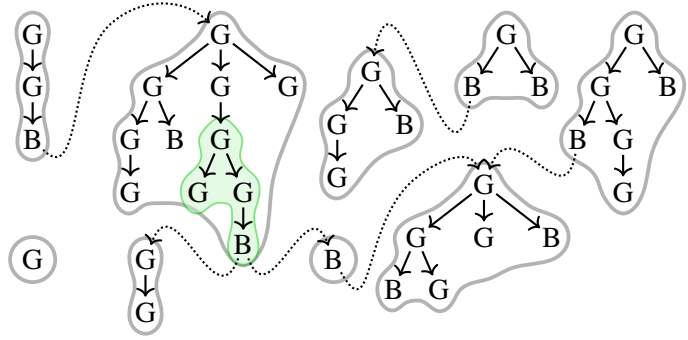

    \centering
    \include{sections/images/zoning}
    \caption{A zoning of a cycle-free graph with three components. Zone borders (gray), good (G) and bad (B) nodes, bridges (dotted) and a match (green) are indicated.}
    \label{fig:zoning:example}
\end{figure}

Figure~\ref{fig:zoning:example} is an abstract depiction of a zoning, and exemplifies the properties established thus far.

\store{proposition}{Bad Node Labels Are Irrelevant}{lem:irrelevance:inert:labels}{
    Let $G^{[\lbl(v) := l]}$ denote the graph obtained by changing  the label of $v \in V_G$ to $l \in \labels$. If $v \notin V_G$, $G^{[\lbl(v) := l]} = G$.
    
    For bad $v \in V_G$ and any $l \in \labels$, if $G \to H$ is a rewrite step via a translated TRS rule $\rho^\enc$ and adherence morphism $\alpha$, then $G^{[\lbl(v) := l]} \to H^{[\lbl(v) := l]}$ is a rewrite step via $\rho^\enc$ and $\alpha$.
}

\begin{proof}
  In a rewrite step, bad nodes are either matched by variable heads, or lie outside the image of $t_L$. In both cases, the label does not influence the application condition, since any label $l$ with $\bot \leq l \leq \top$ is allowed. Moreover, the node is either preserved (and its label unchanged), or deleted. In either case the statement holds.
\end{proof}

\store{theorem}{}{thm:terminating:terms:iff:terminating:graphs}{
Let $R$ be a linear TRS.
    $R$ is terminating on $\ter$ iff $R^\enc$ is terminating on
    {\FinGraphLattice{\Sigma^\enc}}.
}

\begin{proof}
    Direction $\Longleftarrow$ is Lemma~\ref{lemma:terminating:on:graphs:terminating:on:terms}.
    
    For direction $\Longrightarrow$, we prove the contrapositive. By Corollary~\ref{corollary:terminating:iff:cycle-free:terminating}, we may assume $G$ is cycle-free, and thus acyclic. So suppose $R^\enc$ admits an infinite rewrite sequence $\tau_G = G \to G' \to \cdots$ rooted in a cycle-free, finite graph $G$.
    
    Because matches respect zone boundaries, the number of zones is finite, and zones are never created by rewrite steps, there exists a zone $Z$ of $G$ in which a match is fixed and rewritten infinitely often. This zone is at no point affected by matches in other zones, since zones can only affect other zones by completely deleting them. Similarly, due to cycle-freeness, it is easy to see that the bridges and zones connected to $Z$ do not affect rule applicability in $Z$. Hence we can restrict $G$ to $Z$, and construct an infinite rewrite sequence $\tau_Z = Z \to Z' \to Z'' \to \cdots$.
    
    By relabeling every bad node of starting term $Z$ with $\bot$, the existence of an infinite rewrite sequence is not disturbed using Proposition~\ref{lem:irrelevance:inert:labels}. Furthermore, $Z$ is now isomorphic to a term encoding $\term{t}^\enc$ for some term $\term{t}$ (Proposition~\ref{prop:zone:to:term:encoding}). Using the fact that the encoding is closed  (Lemma~\ref{lem:encoding:is:closed}) and that rewriting is defined modulo isomorphism, we can obtain an infinite rewrite sequence on terms. Thus $R$ is also not terminating.
\end{proof}

\begin{remark}
\label{remark:nolte}
Our result may be compared to one due to Nolte~\cite[Chapter 6]{nolte2019}. Nolte first defines two encodings of TRSs into term graph rewriting systems, a basic encoding and an extended encoding. These encodings preserve neither termination nor confluence, and are not embeddings. He then shows that for term graph systems obtained by the basic encoding, there exists a globally termination-preserving encoding into graph rewriting systems (DPO)~\cite[Theorem 6.3]{nolte2019}. So although Nolte's approach is similar to ours in spirit, it does not constitute a globally termination-preserving embedding of TRSs into graph rewriting systems.
\end{remark}

\newcommand{\notimplies}{\;\not\!\!\!\implies}

\begin{remark}[Confluence]
\label{rem:confluence:fails}
    Although $\Longleftarrow$ of Theorem~\ref{thm:terminating:terms:iff:terminating:graphs} holds for confluence as well, $\Longrightarrow$ does not, even if graphs are assumed to be connected, cycle-free and well-labeled. Namely, consider $\Sigma = \{ f,g,h,a,b \}$ with $\#(f) = \#(g) = \#(h) = 1$ and $\#(a) = \#(b) = 0$, and the confluent TRS $R = \{ g(x) \to a , h(x) \to b \}$. Then for the graph
    \begin{tikzcd}[column sep=3mm]
    g \arrow[r, "1"] & f \arrow[r, "1"] & a & f \arrow[l, "1"'] & h \arrow[l, "1"']
    \end{tikzcd}
    both $a$ and $b$ are $R^\enc$-normal forms.
    
    If graphs may be disconnected, rule $g(x) \to a$ even constitutes a counter-example by itself. For the type graph of its rule encoding, a disjoint component $H$ can either be mapped onto the upper context closure (preserving $H$) or the lower context closure (deleting $H$).
\end{remark}

\section{Discussion}
\label{sec:discussion}

We have defined an encoding of linear term rewriting into \pbpostrong rewriting that is both an embedding and globally termination-preserving. These properties are achievable because a \pbpostrong rule allows (i)~specifying where parts of a context may occur around a pattern, (ii)~ensuring that these parts are disjoint, and (iii)~deleting such parts (in our case study, such parts correspond to variable substitutions).

We submit that a rewriting framework $\mathcal{F}$ can be said to be a proper generalization of some other framework $\mathcal{G}$ if there exists an embedding $\encoding$ from $\mathcal{G}$ to $\mathcal{F}$. In this sense, \pbpostrong is a proper generalization of linear term rewriting (and DPO is not). 
Often we want the encoding $\encoding$ to have additional properties such as the global preservation of certain properties (e.g., termination).
For instance, the embedding that interprets the TRS rule $\rho = a(b(x)) \to b(a(x))$ as a mere swap of symbols, and thus as applicable in any context, is an embedding that does not preserve termination globally.
(Note that such an alternative embedding is also expressible in \pbpostrong.)

The fact that a certain property-preserving embedding is possible is an interesting expressiveness result for the embedding formalism. Moreover, it opens up a path to reduction arguments, as was also considered by Nolte~\cite{nolte2019} in a different setting (Remark~\ref{remark:nolte}). In our case, if a \pbpostrong rewrite system is (isomorphic to) the encoding of a TRS (as defined in Definition~\ref{def:translate:trs:rule}), termination can be decided by considering the decoded TRS and forgetting about the complexities of graphs.
Our proof technique extends to more general \pbpostrong rewrite systems as long as the following conditions are met: the pattern of the rules is tree-like (possibly with loops on the nodes of the pattern), the outgoing edges of nodes in the pattern have distinct labels, and the `context' and `variable graphs' are disconnected (except through the pattern) and are not duplicated by the rule.

Our provided embedding into \pbpostrong does not preserve confluence globally. As shown in Remark~\ref{rem:confluence:fails}, the key problem is that an assumption true for terms, namely connectedness, does not hold for graphs. For the same reason it is currently impossible to define a termination-preserving embedding of non-right-linear term rewriting into \pbpostrong: whenever a variable is duplicated, it may also lead to the duplication of any number of disjoint components in the graph that are mapped onto the corresponding variable closure. For future work, we intend to investigate extensions of our encoding that do preserve confluence and termination globally even when variables are duplicated.

Adopting a broader perspective, we hope that our encoding contributes to the development of termination techniques for graph rewriting. There have been recent advances in proving termination of graph rewriting; see for instance work by Bruggink et al.~\cite{bruggink2015proving} and Dershowitz et al.~\cite{dershowitz2018gpo}. In~\cite{dershowitz2018gpo}, recursive path orders are generalized from term rewriting to graph transformation by decomposing the graph into strongly connected components and a well-founded structure between them. A difficulty in this approach is that all possible cycles around the pattern of a rule must be considered. We hope that the technique can be extended to \pbpostrong and strengthened by making use of the application conditions that exclude certain cycles around and through the pattern.

Finally, we believe that our result is a step towards modeling other rewriting formalisms such as lambda calculus and higher order rewriting using graph rewriting. These formalisms also rewrite tree structures, and we expect that extensions of our zoning construction will be instrumental for this purpose. Our goal in this respect is to model these systems in such a way that important properties like termination and confluence are preserved globally, while at the same time keeping the modeling overhead minimal (e.g., avoiding auxiliary rules and rewrite steps that increase the length of rewrite sequences).

\subsubsection*{Acknowledgments} We thank anonymous reviewers for useful suggestions and corrections.
Both authors received funding from the Netherlands Organization for Scientific Research (NWO) under the Innovational Research Incentives Scheme Vidi (project.\ No.\ VI.Vidi.192.004).

\nocite{*}
\bibliographystyle{eptcs}
\bibliography{main}

\end{document}